\documentclass[11pt]{article}

\usepackage[T1]{fontenc}
\usepackage{lmodern, microtype}
\usepackage[left=1.3in, right=1.3in, top=1.25in, bottom=1.25in]{geometry}
\usepackage[small]{titlesec}
\usepackage{mathtools,comment,paralist}
\usepackage{mathtools}
\usepackage{epigraph}
\usepackage{bbm}
\usepackage{xcolor}
\usepackage{amssymb}
\usepackage{fancyhdr}

\usepackage{amssymb,amsmath,amsthm,fancyhdr,bbm}
\usepackage{tikz}
\usetikzlibrary{decorations.pathreplacing,arrows.meta}
\usetikzlibrary{decorations.pathreplacing,decorations.pathmorphing}
\usepackage{subcaption}

\usepackage{xcolor}
\usepackage{soul}
\usepackage{hyperref}
\hypersetup{
  colorlinks=true,
  linkcolor=blue,
  citecolor=red
}
\usepackage{natbib}
\usetikzlibrary{shapes.geometric}
\setcitestyle{authoryear,open={(},close={)}}
\makeatletter

\newtheorem{theorem}{Theorem}
\newtheorem{lemma}{Lemma}
\newtheorem*{lemma*}{Lemma}

\newtheorem*{axiom*}{Axiom}
\newtheorem{proposition}{Proposition}

\newtheorem{remark}{Remark}

\newtheorem*{theorem*}{Theorem}

\newtheorem{definition}{Definition}
\newtheorem*{definition*}{Definition}
\usepackage{graphicx}
\usepackage{pstricks, enumerate, pst-node, pst-text, pst-plot}

\newcommand{\R}{\mathbb{R}}

\newcommand{\N}{\mathbb{N}}

\newcommand{\cX}{\mathcal{X}}
\newcommand{\cY}{\mathcal{Y}}

\usepackage{accents}

\def\finset{\mathcal{P}_f}

\def\ee{\mathrm{e}}

\def\cA{\mathcal{A}}

\def\cM{\mathcal{M}}

\def\Pr{\mathbb{P}}

\def\eps{\varepsilon}

\def\dmax{d_{\mathrm{max}}}

\def\Var{\mathrm{Var}}
\def\L{\mathcal{L}}
\def\U{\mathcal{U}}

\DeclareMathOperator*{\argmin}{\arg\min}

\usepackage{xparse}
\DeclareDocumentCommand\Pr{ m g }{\ensuremath{
    {   \IfNoValueTF {#2}
      {\mathbb{P}\left[{#1}\right]}
      {\mathbb{P}\left[{#1}\middle\vert{#2}\right]}%
    }
}}
\DeclareDocumentCommand\E{ m g }{\ensuremath{
    {   \IfNoValueTF {#2}
      {\mathbb{E}\left[{#1}\right]}
      {\mathbb{E}\left[{#1}\middle\vert{#2}\right]}%
    }
}}
\newcommand{\tv}[1]{d_{\mathrm{TV}}({#1})}

\def\ee{\mathrm{e}}

\def\pub{m}
\def\pri{m}

\newcommand{\olya}[1]{{\color{blue}{(\textbf{Olya:} #1)}}}

\title{Local Coordination and the Geometry of Social Networks}

\author{Tom Hutchcroft \and Olga Rospuskova \and Omer Tamuz\thanks{Caltech. Tom Hutchcroft was supported by a National Science Foundation award (DMS-2246494) and a Packard Fellowship for Science and Engineering. Olga Rospuskova was supported by a Citadel fellowship. Omer Tamuz was supported by a National Science Foundation CAREER award (DMS-1944153) and by a MURI award (N000142412742). We would like to thank Mira Frick, Ryota Iijima and G\'abor Pete for stimulating discussions. }}

\date{\today}

\begin{document}

\maketitle
\begin{abstract}
   We study agents playing a pure coordination game on a large social network. Agents are restricted to coordinate locally, without access to a global communication device, and so different regions of the network will converge to different actions, precluding perfect coordination. We show that the extent of this inefficiency depends on the network geometry: on some networks, near-perfect efficiency is achievable, while on others welfare is strictly bounded away from the optimum. We provide a geometric condition on the network structure that characterizes when near-efficiency is attainable. On networks in which it is unattainable, our results more generally preclude high correlations between outcomes in a large spectrum of dynamic games.
\end{abstract}

\section{Introduction}

Pure coordination problems involve a decision between alternatives that are a priori equivalent, and where the only motivation is to match the choices of others. For example, the choice of a weekly day of rest has obvious coordination incentives, and one can plausibly assume that a priori it does not matter which day of the week it is, assuming everyone chooses the same one. Other examples include the choice of outlet plug standard, the side of the street to drive on, and the term used to describe a new concept or object.

We consider agents who have to coordinate  with their social network neighbors. When the network is large and coordination is local, one cannot hope that the entire population will make the same choice, and so inefficiencies must arise. The question we ask is: on which social networks is it possible to achieve low inefficiency? In particular, what geometric property of the network determines the range of possible outcomes?

We show that low inefficiency is very closely tied to the \emph{amenability}  of the social network graph.\footnote{The term \emph{hyperfiniteness} is often used to describe this property in the math literature; see more in the related literature discussion.} This is a well-studied graph property, which roughly means that the graph can be divided into distinct communities whose size is limited by the locality of interactions, with relatively few edges between communities. We show that on amenable graphs agents can coordinate well in equilibrium (Theorems~\ref{thm:amenable}, \ref{thm:amenable-no-comm} and \ref{thm:amenable-transitive}). On non-amenable graphs, we show that near-efficient coordination is impossible (Theorem~\ref{thm:non-amenable}). Indeed, we show that it is impossible not just in equilibrium, but in any strategy profile.

We consider a setting in which agents can communicate locally: there is some radius of communication $r$, and agents can exchange messages with their neighbors, neighbors of neighbors, etc, up to distance $r$.  This can be viewed as a reduced form model that can capture the end state of a more detailed process, perhaps involving multiple rounds of communication with direct neighbors, which  allows messages to travel some distance.

We initially assume that agents have the ability to send private messages that are not observed by others. On amenable graphs, this suffices for the existence of equilibria with low inefficiency (Theorem~\ref{thm:amenable}). We later consider the case that messages are (locally) public, i.e., observed by all agents up to distance $r$ (e.g., posting a yard sign).  Eliminating private communication still allows for low inefficiency, but with some loss of efficiency (Theorem~\ref{thm:amenable-no-comm}). 

In non-amenable graphs, inefficiency is high, even when private communication is allowed, and indeed even when incentives are not a constraint. Our general result in this setting (Theorem~\ref{thm:non-amenable-general}) is thus also a contribution to the study of stochastic processes on non-amenable graphs, which is an active area of research. %This result has a distributed computing interpretation: no distributed algorithm running on the nodes of a graph using local information can yield highly correlated outputs when the graph is non-amenable. 

The low-inefficiency equilibria we construct on amenable graphs have a simple structure: agents are divided into communities, each community has a leader, and each leader messages the community with instructions determining the action they coordinate on. We call such an equilibrium a leader equilibrium. When communities are insular there are no incentive issues. However, when agents have neighbors who are members of a different community one has to make sure they have no incentive to deviate and follow the other leader's instructions. This issue can be solved by private communication between leaders and their communities (Theorem~\ref{thm:amenable}).  When private communication is not allowed, these issues can be overcome by a careful construction of communities that are sufficiently insular (Theorem~\ref{thm:amenable-no-comm}). These results show that on amenable graphs there are efficient leader equilibria. The converse is also true: the existence of an efficient leader equilibrium implies that the graph is amenable.

We show that inefficiency must be high on non-amenable graphs, and not just in leader equilibria. To this end, we demonstrate that if a graph admits a strategy profile that  achieves low inefficiency, then there is a leader equilibrium that also achieves low inefficiency. This shows that this simple coordination device is (almost) as powerful as any other.

The main technical tool we use to prove this result is what we call the \emph{Shapley influence distribution} \citep{owen2014sobol, owen2017shapley}: a measure of how much a variable affects the outcome of a function that depends on several independent inputs. This distribution is given by the Shapley values of a cooperative game played by the input variables that determine the function's output. We show that the Shapley influence distribution has a Lipschitz property: when applied to two similar functions, it yields similar distributions.

In any equilibrium (or indeed, in any strategy profile) we can view a player's action as a function of the independent inputs given by the messages of the players within the radius of communication. We apply the Shapley influence distribution to this function to yield, for that player, a distribution over the players around them. This will be the distribution used by this player to choose a leader. Intuitively, this means that in the constructed leader equilibrium players are more likely to choose as leaders those in their neighborhood that are more influential to their decisions. The Lipschitz property ensures that agents tend to choose the same leaders as their neighbors, resulting in low inefficiency. 

\paragraph{Related literature.}

The study of coordination has a long history in economics and game theory, including in both cooperative and non-cooperative game theory; a complete survey is beyond the scope of this paper. Coordination on social networks has also been studied extensively, starting with \cite{schelling1969models, schelling1971dynamic}. Perhaps the closest paper to ours is \cite{morris2000contagion}, who studies contagion in a two-action, repeated local interaction model. Agents live on an infinite graph and repeatedly best-respond to their neighbors' actions, myopically maximizing payoffs in a coordination game. The question is: starting from a finite region in which all agents play one action, when will best responses eventually make that action spread to the entire graph? A key message is that this depends on the graph's geometry: slow neighborhood growth supports contagion, reflecting the amenability intuition that communities should have relatively small boundaries for efficient coordination. 

Earlier, \cite{ellison1993learning} studied a similar two-action coordination problem, but in a large finite population under a noisy best-response dynamics. He highlights that coordination outcomes depend on the interaction structure. When coordination is local, the agents converge quickly to the same action, whereas under uniform interactions convergence is extremely slow. A number of other papers follow a similar path; see \cite{weidenholzer2010coordination} for a survey of best-response type models of coordination dynamics. This literature is closely tied to the study of interacting particle systems in probability, which originated with the classical Ising model \citep{ising1925beitrag} and is very active to this day  \citep[see, e.g., a survey by][]{durrett2025interacting}. Both \cite{morris2000contagion} and \cite{ellison1993learning} focus on the long-term outcomes, after repeated interactions have percolated throughout the network, rather than the local interactions we are interested in.

More recent contributions to the coordination on networks literature include  \cite{sadler2020diffusion}, who analyzes diffusion from a seed, where global spread depends on whether adoption opportunities percolate through a giant connected cluster, captured by a branching-factor condition. \cite{chwe2000communication} shows that in a collective action setting, coordination hinges on the network's ability to generate sufficiently strong higher-order beliefs, via a hierarchical communication structure. In his geometric threshold model, dense local overlap facilitates coordination in low-dimensional networks, while weak ties become important mainly once the network is very dense. \cite{pkeski2025random} studies random-threshold coordination and shows that in large networks equilibria admit a simple cutoff structure. Its existence is guaranteed under a small influence condition, where each agent's payoff is not driven by any single neighbor.

Beyond coordination, there is a very large literature that studies how the geometry of the social network affects outcomes in various strategic settings. A very partial list of examples includes \cite{jackson2007diffusion}, who study how network structure affects the diffusion of behavior; \cite{ballester2006s} and \cite{bramoulle2014strategic}, who consider quadratic games with linear best-responses, and show that spectral properties of the adjacency matrix determine equilibrium behavior; and \cite{golub2012homophily}, \cite{acemoglu2011bayesian} and \cite{mossel2015strategic}, who study how network geometry affects social learning outcomes.

Several inequivalent notions of amenability have appeared in the mathematics literature. The notion we use is equivalent to so-called \emph{hyperfiniteness} as introduced by \cite{elek2006combinatorial}; see also \cite{MR2372897}. This notion of (non)amenability could also reasonably be referred to as \emph{statistical (non)amenability} as it concerns a property of graphs determined by the \emph{average} geometry at a random vertex and is robust to the presence of rare outlying regions with atypical geometry. This notion of amenability has been studied extensively by mathematicians in the setting of \emph{unimodular random rooted graphs} \citep{MR2354165} and \emph{measured equivalence relations} \citep{kechris2004topics}, the latter being a central topic in modern descriptive set theory.

The concept of amenability was first introduced by \cite{neumann1929allgemeinen}, who originally defined it as a property of groups, with the goal of understanding the origins of the Banach-Tarski paradox. Classically, a graph is said to be amenable if it contains finite sets of vertices with arbitrarily small surface-to-volume ratio; for groups, being amenable in the sense of von Neumann is equivalent to having Cayley graphs that are amenable in this sense. Finite graphs are always amenable in this sense since the entire vertex set has no boundary. The most classical non-vacuous notion of non-amenability in the finite context is that of being an \emph{expander graph}: Given $\varepsilon>0$, a finite graph is said to be an $\varepsilon$-expander if every set containing at most half the vertices has surface-to-volume ratio at least $\varepsilon$. Such finite expander graphs have many applications in computer science and algorithm design \citep[see, e.g.,][]{hoory2006expander}. Being an expander graph might reasonably be referred to as \emph{uniform non-amenability} in contrast to our statistical notion, and is a strictly stronger property than being non-amenable in our sense. For example, the giant cluster of a random graph on $n$ vertices with average degree $d>1$ is not an expander with high probability when $n$ is large due to the presence of rare ``bad regions'', but is non-amenable in our statistical sense. The same is true for many ``small-world'' models of social networks \citep{watts1998collective}, making our (statistical) notion of non-amenability the more relevant concept in many applications. The two notions are equivalent for transitive graphs \citep{benjamini1999group}, that is, graphs in which any two vertices are related by a symmetry.
Analogues of our main probabilistic theorem (Theorem~\ref{thm:non-amenable-general}) for transitive graphs have been proven using spectral methods in e.g.\  \cite{backhausz2017spectral} and \cite{hutchcroft2023continuity}, but these proof methods do not apply under our notion of (statistical) non-amenability. \cite*{csoka2020entropy} prove a version of this theorem for general (non-transitive) graphs, under a uniform non-amenability assumption (see Remark~\ref{rem:virag}).

% The definition was naturally adapted to infinite graphs \cite{folner?}, and later to finite graphs \cite{?} where it is known as hyperfiniteness \cite{?}... Relation to  descriptive set theory \cite{?} and to distributed computing \cite{?}

\section{Model}

Let $N$ be a finite set of agents. The agents are connected by a social network whose graph is $G=(N,E)$, where $(i,j) \in E \subseteq N \times N$ signifies that  $i$ is a neighbor of $j$. We assume that the graph is undirected, i.e., $(i,j) \in E$ if and only if $(j,i) \in E$. By convention, we set $(i,i) \not  \in E$. We denote the neighborhood of $i \in N$ by $N_i = \{j \,:\, (i,j) \in E\}$. The degree of agent $i$ is the size of $N_i$, and we denote by $\dmax$ the maximum degree of all agents. We think of $\dmax$ as being small: each agent has many fewer neighbors than there are agents in the network.

We say that $i$ is connected to $j$ if there is a sequence of agents starting in $i$ and ending in $j$ such that each subsequent pair are in each other's neighborhood. The connected component of $i$ is the set of agents to whom $i$ is connected. We denote by $B_r(i) \subseteq N$ the set of agents who are at distance at most $r$ from $i$ in the social network. I.e., the agents to whom $i$ is connected by a sequence of agents of length at most $r$. We refer to this set as the $i$'s $r$-neighborhood. 

Each agent has to choose an action in $A = \{-1,+1\}$. The utility of player $i$ for an action profile $a \in A^N$ is given by
\begin{align}
\label{eq:one-shot}
    u_i(a) = -\sum_{j \in N_i}| a_i - a_j|.
\end{align}

This is a pure coordination game in which each agent wants to match its neighbors, receiving a payoff of $0$ for each match and a payoff of $-2$ for each mismatch. All of our results extend to similar coordination games with more actions; we opt to focus on the two action case for simplicity of notation and exposition.

We would like to capture the idea that in order to coordinate, agents can communicate locally and pass information among themselves, but that information cannot travel too far in the network. To this end, we allow for a round of local communication, after which agents choose their actions. 

In the communication stage, each agent $i$ can send each agent $j\in B_r(i)$ a private
message $\pri_{i,j}$. Messages sent
directly to agents within distance $r$ can succinctly model a more complicated
process, e.g., one in which some message from agent $i$ is communicated to the
agents in $B_r(i)$ through chains of exchanges between direct neighbors.

The set of possible messages is denoted by $\cM$. We allow $\cM$ to be large (uncountable), placing no restriction on the complexity of the message. For reasons that will become clear soon, we put a particular structure on $\cM$: messages are infinite strings over a finite alphabet that includes the two actions $\{-1,+1\}$ as letters, as well as the letter $\emptyset$.\footnote{Let $\cA$ be a finite alphabet that includes $\{-1,+1\}$ and $\emptyset$. Let $\cM = \cA^{\N}$, so that a message is an infinite string written using the alphabet $\cA$. Given a letter $\alpha \in \cA$, we identify the string $(\alpha,\emptyset,\emptyset,\ldots)$ with $\alpha$ itself. In particular, $(-1,\emptyset,\emptyset,\ldots)$ is the message $-1$, and likewise  $(+1,\emptyset,\emptyset,\ldots)$ is the message $+1$. The message $(\emptyset,\emptyset,\ldots)$ is the message $\emptyset$, which we call the empty message.}

Agents send private messages simultaneously. After receiving them, they choose their actions simultaneously and receive payoffs according to \eqref{eq:one-shot}. Formally, there are three periods, in which each agent $i$
\begin{enumerate}
    \item sends a private message $\pri_{i,j}$ to each  $j \in B_r(i)$;
    \item learns the private messages $\pri_{\rightarrow i} = (\pri_{j,i})_{j \in B_r(i)}$;
    \item chooses an action $a_i$.
\end{enumerate}
This defines an extensive-form game with simultaneous moves. A pure strategy of agent $i$ is a tuple $\sigma_i = (\pri_{i\rightarrow}, o_i)$, where  $\pri_{i\rightarrow}=(\pri_{i,j})_{j \in B_r(i)}$ are $i$'s private messages, and $o_i$ is a measurable map that assigns to every realization of received messages $\pri_{\rightarrow i}$ an action in $\{-1,+1\}$. We use the same notation to denote mixed strategies, in which case strategies $(\pri_{i\rightarrow}, o_i)_{i \in N}$ will be independent random variables, all defined on the same probability space. Within this probability space, we will denote by $a_i=o_i(\pri_{\rightarrow i})$ the (random) action taken by $i$.\footnote{In the constructions below, we sometimes describe strategies behaviorally, by specifying
randomization at each information set; throughout, these descriptions should
be understood as the equivalent mixed strategies over complete contingency plans.}

We will be interested in mixed equilibria of this game. In particular, we will be interested in the average welfare of this game, i.e., in 
\begin{align*}
    W = -\frac{1}{|E|}\sum_{(i, j)\in E}|a_i- a_j|.
\end{align*}
This is the payoff averaged over all ordered pairs $(i,j)$ of agents who are neighbors. It will be more convenient to study the average inefficiency
\begin{align*}
    I = -W = \frac{2}{|E|}\sum_{(i, j)\in E} \mathbbm{1}_{a_i\neq a_j},
\end{align*}
that is, twice the fraction of pairs of neighbors that mismatch. 
This is equal to $0$ if all agents coordinate, and to $2$ if all agents miscoordinate. 

For any graph, there are two equilibria that achieve first-best: the one in which agents ignore all information and all choose the action $+1$, and another in which they all choose $-1$. We are interested in how agents might coordinate, rather than in the fact that these are stable courses of action; assuming that agents play one of these equilibria relegates the question of coordination on actions to that of coordination on equilibria. Instead, we would like to model a setting in which the agents need to interact and exchange information in order to coordinate. 

To this end, we will define \emph{action-symmetric} strategy profiles. A strategy is action-symmetric if it is invariant to the operation that renames the two actions in the strategy. Playing an action-symmetric strategy captures an a priori indifference to the two actions, making coordination an interesting problem. We will not restrict players to play action-symmetric strategies, but instead study equilibria in which players have no incentive to deviate away from action-symmetry.\footnote{Formally, we define the negation operator on the message alphabet $\cA$ by having it map $-1$ to $+1$ and vice versa, and mapping each other letter to itself. This operator is extended in the obvious way to the set of messages $\cM=\cA^{\N}$, so that the negation of $m = (\alpha_1,\alpha_2,\ldots)\in \cM$ is $-m=(-\alpha_1,-\alpha_2,\ldots) \in \cM$.  A tuple of messages (e.g., a message received from each neighbor) is negated in the same way. \par Let $\cX,\cY$ be spaces on which negation is defined (e.g., $\{-1,+1\},\cA,\cM,B_r(i)^{\cM}$). Let $f \colon X \to Y$ be a function (in particular, we will be interested in $o_i$). We define the  function $\iota(f) \colon X \to Y$ by $[\iota(f)](x)=-f(-x)$. Given a pure strategy $\sigma_i = (\pri_{i\rightarrow},o_i)$, we define  $\iota(\sigma_i) = (-\pri_{i\rightarrow},\iota(o_i))$. We say that $i$'s mixed strategy is action-symmetric if it is $\iota$-invariant: $\sigma_i$ has the same distribution as $\iota(\sigma_i)$. Equivalently, if we denote by $A_i$ the set of strategies of $i$, then $\iota$ is a bijection from $A_i$ to $A_i$, and a probability measure $\mu_i$ over $A_i$ (i.e., a mixed strategy for player $i$) is action-symmetric if for any measurable $E \subseteq A_i$ it holds that $\mu_i(E) = \mu_i(\iota(E))$. }

In a one-shot game without communication, the only action-symmetric equilibrium is the one in which all agents choose each action with probability one half. Indeed, this is the only action-symmetric strategy profile. In the extensive form game with communication, agents can use messages to break ties and coordinate without a prior preference to either action. 

Note that if all players use action-symmetric strategies, then the probability that $a_i=1$ is one-half for all players $i$. In fact, this would suffice as an assumption for all our results: namely, we can replace the action-symmetry assumption with the assumption that strategy profiles are such that each agent takes each action with probability one half. We adhere to our definition as it is a direct assumption on each agent's strategy rather than a joint assumption about the whole strategy profile.

\begin{comment}

Let $\cX,\cY$ be spaces on which negation is defined (e.g., $\{-1,+1\},\cA,\cM,B_r(i)^{\cM}$).
Let $f \colon X \to Y$ be a function (in particular, we will be interested in $o_i$). We define the  function $\iota(f) \colon X \to Y$ by $[\iota(f)](x)=-f(-x)$. For example, suppose $X=\{-1,+1,\emptyset\}$ and $Y=\{-1,+1\}$, so that we can think of $f$ as an action rule that assigns an action to a received message. If $f$ is the constant $+1$, then $\iota(f)$ is the constant $-1$. More generally, if $f(\emptyset)=+1$, then $[\iota(f)](\emptyset)=-1$. If $f(\alpha)=\alpha$ for $\alpha\in\{-1,+1\}$, then $[\iota(f)](\alpha)=f(\alpha)$ for such $\alpha$; likewise, if $f(\alpha)=-\alpha$ for $\alpha\in\{-1,+1\}$, then $[\iota(f)](\alpha)=f(\alpha)$ for such $\alpha$.

Given a pure strategy $\sigma_i = (\pri_{i\rightarrow},o_i)$, we define  $\iota(\sigma_i) = (-\pri_{i\rightarrow},\iota(o_i))$. We say that $i$'s mixed strategy is action-symmetric if it is $\iota$-invariant: $\sigma_i$ has the same distribution as $\iota(\sigma_i)$. Equivalently, if we denote by $A_i$ the set of strategies of $i$, then $\iota$ is a bijection from $A_i$ to $A_i$, and a probability measure $\mu_i$ over $A_i$ (i.e., a mixed strategy for player $i$) is action-symmetric if for any measurable $E \subseteq A_i$ it holds that $\mu_i(E) = \mu_i(\iota(E))$. Intuitively, action-symmetry is invariance to the renaming of the two actions, capturing the notion that players are agnostic to the actions' names.
\end{comment}

Coordination is possible even when all players use action-symmetric strategies. For a simple example, suppose that $r$ is equal to the diameter of the graph (i.e., $j \in B_r(i)$ for all $i,j \in N$), so that every player can communicate with every player. Then there is an equilibrium in which players coordinate perfectly: fix a player $\ell \in N$ whom we call the ``leader''. The leader $\ell$ sends the same private message $\pri_{\ell,i}$, chosen uniformly at random from $\{-1,+1\}$, to every other agent. The leader takes this action, and every other agent follows the leader by choosing $a_i=\pri_{\ell,i}$.  Thus, our communication structure allows for perfect coordination when $r$ is large enough. The question becomes more interesting when $r$ is much smaller than the extent of the graph. 

Regardless of the graph, there is always an equilibrium in which the agents ignore all information and choose $a_i$ uniformly at random. In this case, the expected inefficiency will be high (it will, in fact, be equal to one). The main question of this paper is the following: Given $\eps>0$ and a radius of communication $r$, for which graphs is there an equilibrium that achieves expected inefficiency at most $\eps$?

\subsection{Amenable graphs}
\label{sec:amenable}
The answer to this question turns out to be closely related to a property of graphs called \emph{amenability}. In this section we initially discuss amenability of infinite graphs, where it is simpler to define, before returning our focus to finite graphs. 

To explain this idea, we first consider our game on the line graph, where $N = \mathbb{Z}$ and $j \in N_i$ if $|i-j|=1$ (see Figure~\ref{fig:Z-1d}). We fix some large radius of communication $r$. Notice that it is impossible to achieve perfect coordination: if $i$ and $j$ are more than $2r$ apart, their actions must be statistically independent of each other in any action-symmetric profile, and so it is impossible that all agents take the same action. Nevertheless, it is possible to achieve low inefficiency. 

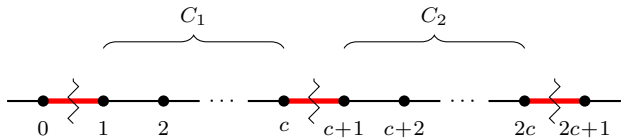
\begin{figure}[h]
        \centering
  \begin{tikzpicture}[x=0.8cm,       
    vertex/.style={circle,fill=black,inner sep=1.5pt},
 font=\scriptsize  
]

% main vertices
\coordinate (m1)      at (0,0);   % -1
\coordinate (zero)    at (1,0);   % 0
\coordinate (one)     at (2,0);   % 1
\coordinate (rminus1) at (4,0);   % r-1
\coordinate (r)       at (5,0);   % r
\coordinate (rplus1)  at (6,0);   % r+1
\coordinate (tworm1)  at (8,0);   % 2r-1
\coordinate (twor)    at (9,0);   % 2r

% small extra segments before -1 and after 2r
\coordinate (L) at (-0.6,0);
\coordinate (R) at (9.6,0);

% symmetric break points around the dots
\coordinate (aL) at (2.6,0);
\coordinate (aR) at (3.4,0);
\coordinate (bL) at (6.6,0);
\coordinate (bR) at (7.4,0);

% boundary edges in red
\draw[line width=2pt,red] (m1) -- (zero);
\draw[line width=2pt,red] (rminus1) -- (r);
\draw[line width=2pt,red] (tworm1) -- (twor);

% black segments (including stubs at ends), with symmetric gaps
\draw[thick] (L) -- (m1);
\draw[thick] (zero) -- (one);
\draw[thick] (one) -- (aL);
\draw[thick] (aR) -- (rminus1);

\draw[thick] (r) -- (rplus1);
\draw[thick] (rplus1) -- (bL);
\draw[thick] (bR) -- (tworm1);
\draw[thick] (twor) -- (R);

% dots centered in the gaps
\node[draw=none,fill=none] at (3,0) {$\cdots$};
\node[draw=none,fill=none] at (7,0) {$\cdots$};

% vertices
\node[vertex] at (m1)      {};
\node[vertex] at (zero)    {};
\node[vertex] at (one)     {};
\node[vertex] at (rminus1) {};
\node[vertex] at (r)       {};
\node[vertex] at (rplus1)  {};
\node[vertex] at (tworm1)  {};
\node[vertex] at (twor)    {};

% labels
\node[below=4pt] at (m1)      {$0$};
\node[below=4pt] at (zero)    {$1$};
\node[below=4pt] at (one)     {$2$};
\node[below=4pt] at (rminus1) {$c$};
\node[below=4pt] at (r)       {$c\!+\!1$};
\node[below=4pt] at (rplus1)  {$c\!+\!2$};
\node[below=4pt] at (tworm1)  {$2c$};
\node[below=4pt] at (twor)    {$2c\!+\!1$};

% braces for communities
\draw[decorate,decoration={brace,amplitude=5pt}]
  (1,0.7) -- (4,0.7)
  node[midway,above=5pt] {$C_1$};

\draw[decorate,decoration={brace,amplitude=5pt}]
  (5,0.7) -- (8,0.7)
  node[midway,above=5pt] {$C_2$};

% vertical zigzag "breakers" on boundary edges
\draw[thin,decorate,decoration={zigzag,segment length=8pt,amplitude=2pt}]
  (0.5,-0.3) -- (0.5,0.3);   % between -1 and 0

\draw[thin,decorate,decoration={zigzag,segment length=8pt,amplitude=2pt}]
  (4.5,-0.3) -- (4.5,0.3);   % between r-1 and r

\draw[thin,decorate,decoration={zigzag,segment length=8pt,amplitude=2pt}]
  (8.5,-0.3) -- (8.5,0.3);   
\end{tikzpicture}      
\captionsetup{width=0.9\textwidth}
        \caption{$\mathbb{Z}$, amenable. Communities have length $c=2r+1$. Miscoordination occurs only along the red edges, which form a small fraction of the edges.}
        \label{fig:Z-1d}
\end{figure}

To this end, we divide the line into communities of length $c=2r+1$: $C_1 = \{1,\ldots,c\}$, $C_2 = \{c+1,\ldots,2c\}$, etc. The leader of $C_k$ will be $\ell_k = k(2r+1)-r$, the agent in the middle of the interval who is within the distance $r$ to all agents in $C_k$. This agent will send the same private message to the entire community and will take the action equal to this message. The agents who are not leaders will send empty messages and will choose the action equal to their leader's message. 

Agents who are within the same community will choose the same action, and so miscoordination will only occur on the boundary of the communities. This is an equilibrium: agents who are not on the boundary will lose payoff if they deviate, as both of their neighbors will take the same recommended action. Agents who are on the boundary of a community likewise will face a lower expected payoff from deviating, since they know which action their community neighbor will take but do not know which action will be taken by their neighbor across the community border.

Since miscoordination occurs only on the community borders, and since there it occurs with probability one half, the inefficiency will be $2/(2c) = 1/(2r+1)$. In particular, when $r$ is large, the inefficiency is small. Clearly, on a finite graph that within radius $r$ looks like this graph (i.e., a cycle of length much longer than $r$) the same conclusion applies, with low inefficiency for large $r$.

A similar equilibrium can be constructed on the two dimensional grid (see Figure~\ref{fig:Z2-2d}). Here, the communities will be $(r+1)\times(r+1)$ squares, where $r$ is even, and therefore have radius $r$  (the radius of a community is the minimum $r$ such that the entire community is contained in the ball of radius $r$ of some member of the community), and again inefficiency will be small when $r$ is large (this time $1/(r+1)$, so again of order $1/r$). As with the line, the same conclusion applies to finite graphs that locally look like the grid.

\begin{figure}[h!]
        \centering
\tikzset{every picture/.style={line width=0.75pt}}  
\tikzset{cutedge/.style={line width=1.2pt, red}}     
\begin{tikzpicture}[x=0.75pt,y=0.75pt,yscale=-0.9,xscale=0.9,font=\scriptsize]
%Shape: Grid [id:dp6414096588376714] 
\draw  [draw opacity=0] (246.67,256.67) -- (246.69,372.35) -- (129,372.37) -- (128.98,256.69) -- cycle ;
\draw  [color={rgb, 255:red, 200; green, 200; blue, 200 }  ,draw opacity=1 ]
 (246.67,256.67) -- (128.98,256.69)
 (246.67,276.67) -- (128.98,276.69)
 (246.68,296.67) -- (128.99,296.69)
 (246.68,316.67) -- (128.99,316.69)
 (246.68,336.67) -- (128.99,336.69)
 (246.69,356.67) -- (129,356.69) ;
\draw  [color={rgb, 255:red, 200; green, 200; blue, 200 }  ,draw opacity=1 ]
 (246.67,256.67) -- (246.69,372.35)
 (226.67,256.67) -- (226.69,372.35)
 (206.67,256.68) -- (206.69,372.35)
 (186.67,256.68) -- (186.69,372.36)
 (166.67,256.68) -- (166.69,372.36)
 (146.67,256.69) -- (146.69,372.36) ;

%Shape: Rectangle [id:dp36081915861183256] 
\draw  [color={rgb, 255:red, 255; green, 255; blue, 255 }  ,draw opacity=1 ][fill=white] (121.18,285.68) -- (276.18,285.68) -- (276.18,307.68) -- (121.18,307.68) -- cycle ;
%Shape: Rectangle [id:dp5451696955868162] 
\draw  [color={rgb, 255:red, 255; green, 255; blue, 255 }  ,draw opacity=1 ][fill=white] (216.99,239.08) -- (218.37,394.08) -- (196.37,394.27) -- (194.99,239.28) -- cycle ;

% red cut edges (right outer, bottom outer, etc.)
\draw[cutedge]   (365.67,256.33) -- (382.67,256.69) ;
\draw[cutedge]   (365.67,276.33) -- (382.67,276.69) ;
\draw[cutedge]   (365.68,316.33) -- (382.68,316.69) ;
\draw[cutedge]   (365.68,336.33) -- (382.68,336.69) ;
\draw[cutedge]   (364.68,356.33) -- (381.68,356.69) ;

\draw[cutedge]   (146.69,356.69) -- (147,374.37) ;
\draw[cutedge]   (166.69,356.68) -- (167,374.36) ;
\draw[cutedge]   (186.69,356.68) -- (187,374.36) ;
\draw[cutedge]   (226.69,356.67) -- (227,374.35) ;
\draw[cutedge]   (246.69,356.67) -- (247,374.35) ;

%Shape: Grid [id:dp043163328109357635] 
\draw  [draw opacity=0] (246.67,139.67) -- (246.69,255.35) -- (129,255.37) -- (128.98,139.69) -- cycle ;
\draw  [color={rgb, 255:red, 200; green, 200; blue, 200 }  ,draw opacity=1 ]
 (246.67,139.67) -- (128.98,139.69)
 (246.67,159.67) -- (128.98,159.69)
 (246.68,179.67) -- (128.99,179.69)
 (246.68,199.67) -- (128.99,199.69)
 (246.68,219.67) -- (128.99,219.69)
 (246.69,239.67) -- (129,239.69) ;
\draw  [color={rgb, 255:red, 200; green, 200; blue, 200 }  ,draw opacity=1 ]
 (246.67,139.67) -- (246.69,255.35)
 (226.67,139.67) -- (226.69,255.35)
 (206.67,139.68) -- (206.69,255.35)
 (186.67,139.68) -- (186.69,255.36)
 (166.67,139.68) -- (166.69,255.36)
 (146.67,139.69) -- (146.69,255.36) ;

% white rectangles clearing labels
\draw  [color=white,draw opacity=1,fill=white] (121.18,168.68) -- (276.18,168.68) -- (276.18,190.68) -- (121.18,190.68) -- cycle ;
\draw  [color=white,draw opacity=1,fill=white] (216.99,122.08) -- (218.37,277.08) -- (196.37,277.27) -- (194.99,122.28) -- cycle ;

% more red cuts (top outer, vertical outer, etc.)
\draw[cutedge]   (365.67,139.33) -- (382.67,139.69) ;
\draw[cutedge]   (365.67,159.33) -- (382.67,159.69) ;
\draw[cutedge]   (365.68,199.33) -- (382.68,199.69) ;
\draw[cutedge]   (365.68,219.33) -- (382.68,219.69) ;
\draw[cutedge]   (364.68,239.33) -- (381.68,239.69) ;

\draw[cutedge]   (146.69,239.69) -- (147,257.37) ;
\draw[cutedge]   (166.69,239.68) -- (167,257.36) ;
\draw[cutedge]   (186.69,239.68) -- (187,257.36) ;
\draw[cutedge]   (226.69,239.67) -- (227,257.35) ;
\draw[cutedge]   (246.69,239.67) -- (247,257.35) ;

%Shape: Grid [id:dp1431789966799044] 
\draw  [draw opacity=0] (364.67,256.67) -- (364.69,372.35) -- (247,372.37) -- (246.98,256.69) -- cycle ;
\draw  [color={rgb, 255:red, 200; green, 200; blue, 200 }  ,draw opacity=1 ]
 (364.67,256.67) -- (246.98,256.69)
 (364.67,276.67) -- (246.98,276.69)
 (364.68,296.67) -- (246.99,296.69)
 (364.68,316.67) -- (246.99,316.69)
 (364.68,336.67) -- (246.99,336.69)
 (364.69,356.67) -- (247,356.69) ;
\draw  [color={rgb, 255:red, 200; green, 200; blue, 200 }  ,draw opacity=1 ]
 (364.67,256.67) -- (364.69,372.35)
 (344.67,256.67) -- (344.69,372.35)
 (324.67,256.68) -- (324.69,372.35)
 (304.67,256.68) -- (304.69,372.36)
 (284.67,256.68) -- (284.69,372.36)
 (264.67,256.69) -- (264.69,372.36) ;

\draw  [color=white,draw opacity=1,fill=white] (239.18,285.68) -- (394.18,285.68) -- (394.18,307.68) -- (239.18,307.68) -- cycle ;
\draw  [color=white,draw opacity=1,fill=white] (334.99,239.08) -- (336.37,394.08) -- (314.37,394.27) -- (312.99,239.28) -- cycle ;

\draw[cutedge]   (247.67,256.33) -- (264.67,256.69) ;
\draw[cutedge]   (247.67,276.33) -- (264.67,276.69) ;
\draw[cutedge]   (247.68,316.33) -- (264.68,316.69) ;
\draw[cutedge]   (247.68,336.33) -- (264.68,336.69) ;
\draw[cutedge]   (246.68,356.33) -- (263.68,356.69) ;

\draw[cutedge]   (264.69,356.69) -- (265,374.37) ;
\draw[cutedge]   (284.69,356.68) -- (285,374.36) ;
\draw[cutedge]   (304.69,356.68) -- (305,374.36) ;
\draw[cutedge]   (344.69,356.67) -- (345,374.35) ;
\draw[cutedge]   (364.69,356.67) -- (365,374.35) ;

%Shape: Grid [id:dp8681643574444563] 
\draw  [draw opacity=0] (364.67,138.67) -- (364.69,254.35) -- (247,254.37) -- (246.98,138.69) -- cycle ;
\draw  [color={rgb, 255:red, 200; green, 200; blue, 200 }  ,draw opacity=1 ]
 (364.67,138.67) -- (246.98,138.69)
 (364.67,158.67) -- (246.98,158.69)
 (364.68,178.67) -- (246.99,178.69)
 (364.68,198.67) -- (246.99,198.69)
 (364.68,218.67) -- (246.99,218.69)
 (364.69,238.67) -- (247,238.69) ;
\draw  [color={rgb, 255:red, 200; green, 200; blue, 200 }  ,draw opacity=1 ]
 (364.67,138.67) -- (364.69,254.35)
 (344.67,138.67) -- (344.69,254.35)
 (324.67,138.68) -- (324.69,254.35)
 (304.67,138.68) -- (304.69,254.36)
 (284.67,138.68) -- (284.69,254.36)
 (264.67,138.69) -- (264.69,254.36) ;

\draw  [color=white,draw opacity=1,fill=white] (239.18,167.68) -- (394.18,167.68) -- (394.18,189.68) -- (239.18,189.68) -- cycle ;
\draw  [color=white,draw opacity=1,fill=white] (334.99,121.08) -- (336.37,276.08) -- (314.37,276.27) -- (312.99,121.28) -- cycle ;

\draw[cutedge]   (247.67,138.33) -- (264.67,138.69) ;
\draw[cutedge]   (247.67,158.33) -- (264.67,158.69) ;
\draw[cutedge]   (247.68,198.33) -- (264.68,198.69) ;
\draw[cutedge]   (247.68,218.33) -- (264.68,218.69) ;
\draw[cutedge]   (246.68,238.33) -- (263.68,238.69) ;

\draw[cutedge]   (264.69,238.69) -- (265,256.37) ;
\draw[cutedge]   (284.69,238.68) -- (285,256.36) ;
\draw[cutedge]   (304.69,238.68) -- (305,256.36) ;
\draw[cutedge]   (344.69,238.67) -- (345,256.35) ;
\draw[cutedge]   (364.69,238.67) -- (365,256.35) ;

% ===== moved bold black community rectangles (only this changed) =====
\draw  [line width=2.25]  (139,250) -- (255,250) -- (255,365.87) -- (139,365.87) -- cycle ;
\draw  [line width=2.25]  (139,131.87) -- (255,131.87) -- (255,250) -- (139,250) -- cycle ;
\draw  [line width=2.25]  (255,131.87) -- (374,131.87) -- (374,250) -- (255,250) -- cycle ;
\draw  [line width=2.25]  (255,250) -- (374,250) -- (374,365.87) -- (255,365.87) -- cycle ;

% remaining red cuts on top/left etc.
\draw[cutedge]   (146.69,122.69) -- (147,140.37) ;
\draw[cutedge]   (166.69,122.68) -- (167,140.36) ;
\draw[cutedge]   (186.69,122.68) -- (187,140.36) ;
\draw[cutedge]   (226.69,122.67) -- (227,140.35) ;
\draw[cutedge]   (246.69,122.67) -- (247,140.35) ;
\draw[cutedge]   (264.69,122.69) -- (265,140.37) ;
\draw[cutedge]   (284.69,122.68) -- (285,140.36) ;
\draw[cutedge]   (304.69,122.68) -- (305,140.36) ;
\draw[cutedge]   (344.69,122.67) -- (345,140.35) ;
\draw[cutedge]   (364.69,122.67) -- (365,140.35) ;

\draw[cutedge]   (129.67,256.33) -- (146.67,256.69) ;
\draw[cutedge]   (129.67,276.33) -- (146.67,276.69) ;
\draw[cutedge]   (129.68,316.33) -- (146.68,316.69) ;
\draw[cutedge]   (129.68,336.33) -- (146.68,336.69) ;
\draw[cutedge]   (128.68,356.33) -- (145.68,356.69) ;
\draw[cutedge]   (129.67,139.33) -- (146.67,139.69) ;
\draw[cutedge]   (129.67,159.33) -- (146.67,159.69) ;
\draw[cutedge]   (129.68,199.33) -- (146.68,199.69) ;
\draw[cutedge]   (129.68,219.33) -- (146.68,219.69) ;
\draw[cutedge]   (128.68,239.33) -- (145.68,239.69) ;

% Text nodes (unchanged)
\draw (107,360) node [anchor=north west] {(1,1)};
%\draw (205,252) node [anchor=north west] {(r/2,r/2)};

\draw (370,115) node [anchor=north west] {($2r+2$,$2r+2$)};
\draw (370,360) node [anchor=north west] {($2r+2$,1)};

\draw (78,115) node [anchor=north west] {(1,$2r+2$)};
\draw (208,297) node {$C_{1}$};
\draw (208,180) node {$C_{2}$};
\draw (327,297) node {$C_{3}$};
\draw (327,180) node {$C_{4}$};

\end{tikzpicture}
\captionsetup{width=0.9\textwidth}
        \caption{$\mathbb{Z}^2$, amenable. Communities are $(r+1)\times(r+1)$ squares. Miscoordination occurs only along the red edges, which, as on the line, form a small fraction of the edges.}

        \label{fig:Z2-2d}
\end{figure}
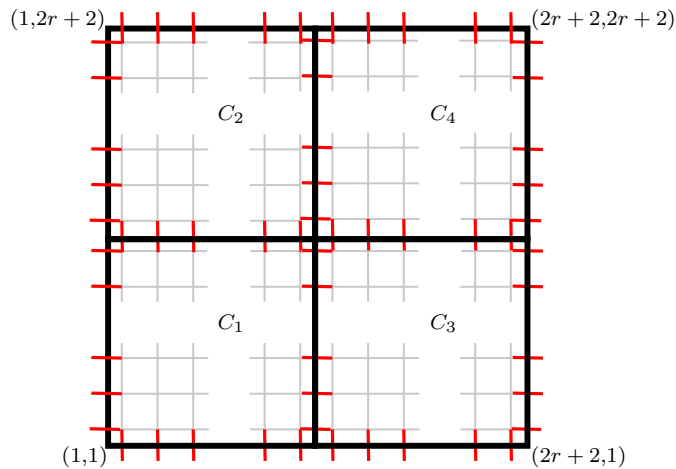

Both of these constructions rely on a geometric property of these graphs: namely, that they can be partitioned into communities in such a way that each community has a small boundary as compared to the number of agents it contains. This is not possible on every graph. For example, consider the (infinite) binary tree (see Figure~\ref{fig:tree}). Here, if we choose as a community, say, the first ten levels, then more than half of the agents will be on the boundary of the community, and so inefficiency will be high. This holds if we choose the community to consist of any number of levels. In fact, this holds for any finite community on this graph: the number of edges between members of the community and members outside the community will be at least the size of the community. Thus, on a finite graph that locally looks like this tree, there cannot be efficient leader equilibria, regardless of how large $r$ is.\footnote{An example of a finite graph that locally looks like a tree is an $(n,d/n)$-Erd\H{o}s-R\'enyi graph for $d>0$, which, for large $n$, with high probability looks like a tree within an $r$-neighborhood of almost all vertices. More formally, if $G_n=(N,E)$ is chosen from this measure over graphs of $n$ vertices, and if $i$ is a uniformly chosen vertex, then the probability that $i$ is in a cycle of length at most $r$ tends to zero as $n$ tends to infinity. The graph will locally look like a Galton-Watson tree, rather than a binary tree.}

\begin{figure}[h]

    \centering 
\tikzset{every picture/.style={line width=0.75pt}} %set default line width to 0.75pt        

\begin{tikzpicture}[x=0.75pt,y=0.75pt,yscale=-1,xscale=1.2]
%uncomment if require: \path (0,444); %set diagram left start at 0, and has height of 444

%Straight Lines [id:da8469959538362166] 
\draw    (449.61,141.66) -- (528.09,162.47) ;
%Straight Lines [id:da9403340839339984] 
\draw    (449.61,141.66) -- (375.31,162.96) ;
%Straight Lines [id:da5700381256391666] 
\draw    (528.09,162.47) -- (563.32,184.92) ;
%Straight Lines [id:da4794302708509276] 
\draw    (528.09,162.47) -- (491.35,185.15) ;
%Straight Lines [id:da054063105715418946] 
\draw    (563.32,184.92) -- (587.75,207.4) ;
%Straight Lines [id:da11786685185256884] 
\draw    (563.31,184.17) -- (542.78,208.3) ;
%Straight Lines [id:da236884649825997] 
\draw    (598.7,239.19) -- (598.9,247.46) ;
%Straight Lines [id:da9319507058001404] 
\draw [color={rgb, 255:red, 255; green, 0; blue, 0 }  ,draw opacity=1 ]   (598.9,247.46) -- (602.71,269.26) ;
%Straight Lines [id:da7775911456150004] 
\draw [color={rgb, 255:red, 255; green, 0; blue, 31 }  ,draw opacity=1 ]   (598.9,247.46) -- (595.37,269.28) ;
%Straight Lines [id:da023081969669879188] 
\draw    (491.35,185.9) -- (515.79,208.39) ;
%Straight Lines [id:da10270080156381967] 
\draw    (491.35,185.15) -- (470.81,209.28) ;
%Straight Lines [id:da7553976647044236] 
\draw    (588.66,208.15) -- (594.1,214.9) ;
%Straight Lines [id:da2244460289188167] 
\draw    (588.66,208.15) -- (583.3,214.94) ;
%Straight Lines [id:da10359947911051015] 
\draw    (543.68,209.05) -- (549.12,215.8) ;
%Straight Lines [id:da1820196417106642] 
\draw    (543.68,209.05) -- (538.33,215.83) ;
%Straight Lines [id:da9960702469536868] 
\draw    (471.7,208.53) -- (477.15,215.28) ;
%Straight Lines [id:da9659892821725642] 
\draw    (471.7,208.53) -- (466.35,215.31) ;
%Straight Lines [id:da23540075568292762] 
\draw    (514.89,209.14) -- (520.33,215.89) ;
%Straight Lines [id:da15355132268640737] 
\draw    (514.89,209.14) -- (509.54,215.93) ;
%Straight Lines [id:da6226086063373447] 
\draw    (584.01,239.24) -- (584.21,247.51) ;
%Straight Lines [id:da44244406490830956] 
\draw [color={rgb, 255:red, 255; green, 0; blue, 0 }  ,draw opacity=1 ]   (584.21,247.51) -- (588.02,269.31) ;
%Straight Lines [id:da011911502652895756] 
\draw [color={rgb, 255:red, 255; green, 0; blue, 31 }  ,draw opacity=1 ]   (584.21,247.51) -- (580.68,269.33) ;
%Straight Lines [id:da8361229551668757] 
\draw    (569.32,239.28) -- (569.52,247.55) ;
%Straight Lines [id:da09916761141587072] 
\draw [color={rgb, 255:red, 255; green, 0; blue, 0 }  ,draw opacity=1 ]   (569.52,247.55) -- (573.33,269.35) ;
%Straight Lines [id:da7598260705940069] 
\draw [color={rgb, 255:red, 255; green, 0; blue, 31 }  ,draw opacity=1 ]   (569.52,247.55) -- (565.99,269.38) ;
%Straight Lines [id:da8010060893978703] 
\draw    (554.63,239.33) -- (554.83,247.6) ;
%Straight Lines [id:da7938534937360819] 
\draw [color={rgb, 255:red, 255; green, 0; blue, 0 }  ,draw opacity=1 ]   (554.83,247.6) -- (558.64,269.4) ;
%Straight Lines [id:da666805586087334] 
\draw [color={rgb, 255:red, 255; green, 0; blue, 31 }  ,draw opacity=1 ]   (554.83,247.6) -- (551.3,269.42) ;
%Straight Lines [id:da7566556448618582] 
\draw    (375.31,162.96) -- (410.54,185.41) ;
%Straight Lines [id:da47360053204684127] 
\draw    (375.31,162.96) -- (338.57,185.64) ;
%Straight Lines [id:da24100333695960818] 
\draw    (410.54,185.41) -- (434.97,207.89) ;
%Straight Lines [id:da9390014060441009] 
\draw    (410.53,184.66) -- (390,208.79) ;
%Straight Lines [id:da08220980728916782] 
\draw    (338.57,186.39) -- (363.01,208.88) ;
%Straight Lines [id:da4600046354671099] 
\draw    (338.57,185.64) -- (318.03,209.77) ;
%Straight Lines [id:da7768108426031236] 
\draw    (435.88,208.64) -- (441.32,215.39) ;
%Straight Lines [id:da47209105443723065] 
\draw    (435.88,208.64) -- (430.52,215.43) ;
%Straight Lines [id:da24054294991129233] 
\draw    (390.9,209.54) -- (396.34,216.29) ;
%Straight Lines [id:da6109927711748606] 
\draw    (390.9,209.54) -- (385.55,216.32) ;
%Straight Lines [id:da7577417343949312] 
\draw    (318.92,209.02) -- (324.37,215.77) ;
%Straight Lines [id:da5931571471889051] 
\draw    (318.92,209.02) -- (313.57,215.8) ;
%Straight Lines [id:da43092420345668403] 
\draw    (362.11,209.63) -- (367.55,216.38) ;
%Straight Lines [id:da863772766970921] 
\draw    (362.11,209.63) -- (356.76,216.42) ;
%Straight Lines [id:da4711006536450655] 
\draw    (509.52,240.23) -- (509.72,248.5) ;
%Straight Lines [id:da33516338447861815] 
\draw [color={rgb, 255:red, 255; green, 0; blue, 0 }  ,draw opacity=1 ]   (509.72,248.5) -- (513.53,270.3) ;
%Straight Lines [id:da0873897350546986] 
\draw [color={rgb, 255:red, 255; green, 0; blue, 31 }  ,draw opacity=1 ]   (509.72,248.5) -- (506.18,270.32) ;
%Straight Lines [id:da8434229588471632] 
\draw    (494.83,240.27) -- (495.03,248.54) ;
%Straight Lines [id:da6813016458692503] 
\draw [color={rgb, 255:red, 255; green, 0; blue, 0 }  ,draw opacity=1 ]   (495.03,248.54) -- (498.84,270.34) ;
%Straight Lines [id:da09167998236662633] 
\draw [color={rgb, 255:red, 255; green, 0; blue, 31 }  ,draw opacity=1 ]   (495.03,248.54) -- (491.49,270.37) ;
%Straight Lines [id:da16492922006193322] 
\draw    (480.14,240.32) -- (480.34,248.59) ;
%Straight Lines [id:da36852578962715354] 
\draw [color={rgb, 255:red, 255; green, 0; blue, 0 }  ,draw opacity=1 ]   (480.34,248.59) -- (484.15,270.39) ;
%Straight Lines [id:da8050152931722941] 
\draw [color={rgb, 255:red, 255; green, 0; blue, 31 }  ,draw opacity=1 ]   (480.34,248.59) -- (476.8,270.41) ;
%Straight Lines [id:da39619886371620994] 
\draw    (465.45,240.37) -- (465.65,248.64) ;
%Straight Lines [id:da684203507552533] 
\draw [color={rgb, 255:red, 255; green, 0; blue, 0 }  ,draw opacity=1 ]   (465.65,248.64) -- (469.46,270.44) ;
%Straight Lines [id:da9100647134198045] 
\draw [color={rgb, 255:red, 255; green, 0; blue, 31 }  ,draw opacity=1 ]   (465.65,248.64) -- (462.11,270.46) ;
%Straight Lines [id:da10385636816470412] 
\draw    (445.51,240.43) -- (445.71,248.7) ;
%Straight Lines [id:da2942729183059075] 
\draw [color={rgb, 255:red, 255; green, 0; blue, 0 }  ,draw opacity=1 ]   (445.71,248.7) -- (449.52,270.5) ;
%Straight Lines [id:da9999321318328588] 
\draw [color={rgb, 255:red, 255; green, 0; blue, 31 }  ,draw opacity=1 ]   (445.71,248.7) -- (442.18,270.52) ;
%Straight Lines [id:da9184756723981866] 
\draw    (430.82,240.48) -- (431.02,248.75) ;
%Straight Lines [id:da579313934124604] 
\draw [color={rgb, 255:red, 255; green, 0; blue, 0 }  ,draw opacity=1 ]   (431.02,248.75) -- (434.83,270.55) ;
%Straight Lines [id:da7670493997606325] 
\draw [color={rgb, 255:red, 255; green, 0; blue, 31 }  ,draw opacity=1 ]   (431.02,248.75) -- (427.49,270.57) ;
%Straight Lines [id:da8358652023232002] 
\draw    (416.13,240.52) -- (416.33,248.8) ;
%Straight Lines [id:da10903454307051097] 
\draw [color={rgb, 255:red, 255; green, 0; blue, 0 }  ,draw opacity=1 ]   (416.33,248.8) -- (420.14,270.59) ;
%Straight Lines [id:da2401636780301285] 
\draw [color={rgb, 255:red, 255; green, 0; blue, 31 }  ,draw opacity=1 ]   (416.33,248.8) -- (412.8,270.62) ;
%Straight Lines [id:da5775386939265975] 
\draw    (401.44,240.57) -- (401.64,248.84) ;
%Straight Lines [id:da3101569273365514] 
\draw [color={rgb, 255:red, 255; green, 0; blue, 0 }  ,draw opacity=1 ]   (401.64,248.84) -- (405.45,270.64) ;
%Straight Lines [id:da8764227912228265] 
\draw [color={rgb, 255:red, 255; green, 0; blue, 31 }  ,draw opacity=1 ]   (401.64,248.84) -- (398.11,270.67) ;
%Straight Lines [id:da9332868167275209] 
\draw    (356.32,241.47) -- (356.52,249.74) ;
%Straight Lines [id:da7420698047182198] 
\draw [color={rgb, 255:red, 255; green, 0; blue, 0 }  ,draw opacity=1 ]   (356.52,249.74) -- (360.34,271.54) ;
%Straight Lines [id:da6442215405262578] 
\draw [color={rgb, 255:red, 255; green, 0; blue, 31 }  ,draw opacity=1 ]   (356.52,249.74) -- (352.99,271.56) ;
%Straight Lines [id:da6817003151229476] 
\draw    (341.63,241.52) -- (341.83,249.79) ;
%Straight Lines [id:da8033358789092806] 
\draw [color={rgb, 255:red, 255; green, 0; blue, 0 }  ,draw opacity=1 ]   (341.83,249.79) -- (345.65,271.59) ;
%Straight Lines [id:da8827078951828061] 
\draw [color={rgb, 255:red, 255; green, 0; blue, 31 }  ,draw opacity=1 ]   (341.83,249.79) -- (338.3,271.61) ;
%Straight Lines [id:da7116122964987543] 
\draw    (326.94,241.56) -- (327.14,249.83) ;
%Straight Lines [id:da0021511723058249554] 
\draw [color={rgb, 255:red, 255; green, 0; blue, 0 }  ,draw opacity=1 ]   (327.14,249.83) -- (330.96,271.63) ;
%Straight Lines [id:da27686082583523697] 
\draw [color={rgb, 255:red, 255; green, 0; blue, 31 }  ,draw opacity=1 ]   (327.14,249.83) -- (323.61,271.66) ;
%Straight Lines [id:da9253595669449427] 
\draw    (312.25,241.61) -- (312.45,249.88) ;
%Straight Lines [id:da13320380827272948] 
\draw [color={rgb, 255:red, 255; green, 0; blue, 0 }  ,draw opacity=1 ]   (312.45,249.88) -- (316.27,271.68) ;
%Straight Lines [id:da7267350503076517] 
\draw [color={rgb, 255:red, 255; green, 0; blue, 31 }  ,draw opacity=1 ]   (312.45,249.88) -- (308.92,271.7) ;

% =========================
% Dots aligned with nodes
% =========================

% Main tree nodes (root + 3 levels)
\foreach \x/\y in {
  449.61/141.66,
  528.09/162.47,
  375.31/162.96,
  563.32/184.92,
  491.35/185.15,
  410.54/185.41,
  338.57/185.64,
  587.75/207.40,
  542.78/208.30,
  515.79/208.39,
  470.81/209.28,
  434.97/207.89,
  390.00/208.79,
  363.01/208.88,
  318.03/209.77
}{
  \draw (\x,\y) node[font=\Huge] {$.$};
}

% Branch points before red splits (at y ≈ 247–250)
\foreach \x/\y in {
  598.90/247.46,
  584.21/247.51,
  569.52/247.55,
  554.83/247.60,
  509.72/248.50,
  495.03/248.54,
  480.34/248.59,
  465.65/248.64,
  445.71/248.70,
  431.02/248.75,
  416.33/248.80,
  401.64/248.84,
  356.52/249.74,
  341.83/249.79,
  327.14/249.83,
  312.45/249.88
}{
  \draw (\x,\y) node[font=\Huge] {$.$};
}

% Red leaf endpoints (both children of each branch)
\foreach \x/\y in {
  602.71/269.26,
  595.37/269.28,
  588.02/269.31,
  580.68/269.33,
  573.33/269.35,
  565.99/269.38,
  558.64/269.40,
  551.30/269.42,
  513.53/270.30,
  506.18/270.32,
  498.84/270.34,
  491.49/270.37,
  484.15/270.39,
  476.80/270.41,
  469.46/270.44,
  462.11/270.46,
  449.52/270.50,
  442.18/270.52,
  434.83/270.55,
  427.49/270.57,
  420.14/270.59,
  412.80/270.62,
  405.45/270.64,
  398.11/270.67,
  360.34/271.54,
  352.99/271.56,
  345.65/271.59,
  338.30/271.61,
  330.96/271.63,
  323.61/271.66,
  316.27/271.68,
  308.92/271.70
}{
  \draw (\x,\y) node[font=\Huge,color={rgb,255:red,255;green,0;blue,0}] {$.$};
}

% Ellipsis markers (kept where you had them)
\draw (536.82,258.90) node[anchor=north west,inner sep=0.75pt,font=\scriptsize,rotate=-179.65] {$...$};
\draw (383.63,258.90) node[anchor=north west,inner sep=0.75pt,font=\scriptsize,rotate=-179.65] {$...$};

\draw (570.82,220.90) node[anchor=north west,inner sep=0.75pt,font=\scriptsize,rotate=-179.65, rotate=90] {$...$};
\draw (495.63,220.90) node[anchor=north west,inner sep=0.75pt,font=\scriptsize,rotate=-179.65, rotate=90] {$...$};

\draw (342.82,220.90) node[anchor=north west,inner sep=0.75pt,font=\scriptsize,rotate=-179.65, rotate=90] {$...$};
\draw (417.63,220.90) node[anchor=north west,inner sep=0.75pt,font=\scriptsize,rotate=-179.65, rotate=90] {$...$};

\draw (304.63,245.14) node[anchor=north west,inner sep=0.75pt,font=\scriptsize] {$1$};

\draw (317.63,245.14) node[anchor=north west,inner sep=0.75pt,font=\scriptsize] {$2$};

\draw (332.0,245.14) node[anchor=north west,inner sep=0.75pt,font=\scriptsize] {$3$};

\draw (602.0,243.14) node[anchor=north west,inner sep=0.75pt,font=\scriptsize] {$2^r$};

\end{tikzpicture}

        \caption{Binary tree, non-amenable.}
        \label{fig:tree}
    
\end{figure}
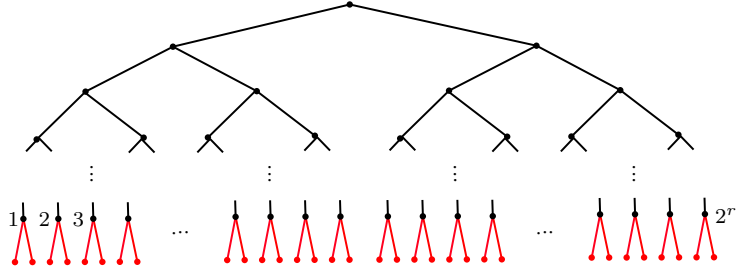

In infinite graphs, the existence of finite communities with small boundaries is captured by the definition of an \emph{amenable graph}. Formally, let $G=(N,E)$ be a graph, and assume that $N$ is countable and each neighborhood $N_i$ is finite. We write $\finset(N)$ for the family of finite subsets of $N$. For $F\in\finset(N)$ define its \emph{boundary}
\[
\partial F \;=\; \{\, (i,j) \in E\,:\, i\in F,\; j\notin F \,\},
\]
i.e., the edges that cross from $F$ to outside of $F$. The \textit{surface-to-volume ratio} of $F$ is $\frac{|\partial F|}{|F|}$. The graph $G$ is said to be \emph{amenable} if it has finite sets with an arbitrarily small surface-to-volume ratio:
\begin{align*}
    \inf_{F \in \finset(N)}\frac{|\partial F|}{|F|} = 0.
\end{align*}
The infinite line and infinite grid are amenable, whereas the infinite binary tree is not. 

In our model graphs are finite, and so we will need a finer definition, as all finite graphs are amenable.\footnote{Let $F = V$. Then $\partial F$ is the empty set, and so $|\partial F|/|F|=0$.} Rather than just requiring the existence of communities with low surface-to-volume ratio, this definition will require a partition of the graph into such sets, so that almost all agents must be members of such communities. We will also need these sets to have radius at most $r$, so that agents can coordinate actions within their communities; we say that $C \subseteq N$ has radius at most $r$ if there is an agent $i \in C$ such that $C$ is a subset of $B_r(i)$.

\begin{definition}
     A finite graph $G=(N,E)$ is $(\eps,r)$-amenable if there exists a partition of $N$ into connected parts of radius at most $r$, and the number of edges connecting nodes in different parts is at most $\eps|E|$.
\end{definition}
Equivalently, there is a set of edges $D \subseteq E$ such that $|D| \leq \eps|E|$, and for each $i \in N$ the connected component of $i$ in $G'=(N,E \setminus D)$ has radius at most $r$, where the radius is measured in the
original graph $G$ (rather than in $G'$). That is, one can remove a small number of connections so that the remaining graph is separated into small communities. Note that while some of these communities may have a large surface-to-volume ratio, it is impossible for these communities to contain many agents, since there are at most $\eps|E|$ edges in the union of all boundaries. Indeed, the average surface-to-volume ratio (weighted by population size) is precisely $|D|/|N|$, which is at most $\eps|E|/|N|$. Assuming each agent is connected to a small number of others, this means that when $\eps$ is small, almost all agents will be in connected components that have small surface-to-volume ratios.\footnote{Formally, if $G$ is $(\eps,r)$-amenable, then at most a $\sqrt{2\eps \dmax}$-fraction of the agents will be in connected components with surface-to-volume ratio greater than $\sqrt{2\eps \dmax}$.} In the math literature, this property is often called \emph{hyperfiniteness} rather than amenability. A perhaps more appropriate name is \emph{statistical amenability}, as it requires that almost all nodes be the members of a set with a small surface-to-volume ratio, rather than just the existence of such sets, as in the original definition of amenability for infinite graphs.

It is easy to see that for any $\eps>0$, every cycle graph will be $(\eps,r)$-amenable, assuming $r$ is large enough. The same holds for an $n$-by-$n$ two dimensional grid. An interesting class of finite graphs which are amenable for non-obvious reasons are the bounded degree finite planar graphs\footnote{A graph is \emph{planar} if it can be drawn on the plane without any of its edges intersecting.}, which are amenable by a theorem of \cite{lipton1977applications,lipton1979separator} in the following sense: for every maximal degree $\dmax$ and $\eps>0$ there is an $r$ such that every planar graph is $(\eps,r)$-amenable. This holds more generally for bounded degree graphs with any excluded minor \citep*{alon1990separator}. An example of a non-amenable graph is an Erd\H{o}s-R\'enyi graph, in which each pair of agents are connected by an edge at random, independently and with some probability $p$. In particular, fixing the expected degree $d = n p > 1$,  there is an $\eps>0$ such that for any $r$ the probability that the graph is $(\eps,r)$-amenable tends to zero as the number of agents $n$ tends to infinity. These graphs locally look like trees, but globally have (almost) no leaves, which is impossible for a planar graph.

\section{Coordination on amenable graphs}
\label{sec:coordination_private}

Our first main result shows that when a graph is $(\varepsilon,r)$-amenable, we can find an equilibrium in which agents coordinate well. These equilibria will be \emph{leader equilibria}, which we informally described above.

Fix $G = (N,E)$ and a partition of $N$ into communities $\{C_1,C_2,\ldots\}$ of radius at most $r$. In each community $C_k$ we choose a leader $\ell_k$ such that $C_k$ is contained in $B_r(\ell_k)$; the existence of $\ell_k$ follows immediately from the fact that $C_k$ has radius at most $r$.

In a leader equilibrium, each community coordinates through its leader's messages. Each leader $\ell_k$ of $C_k$ sends an empty message $\pri_{\ell_k,j}$ to each $j \in B_r(\ell_k)$ who is not in their community. To the members of their community, leaders send the same message, chosen uniformly at random from $\{-1,+1\}$. I.e., if $i,j \in C_k$, then $\pri_{\ell_k,i}=\pri_{\ell_k,j} \in \{-1,+1\}$. Non-leaders send empty messages. Finally, upon observing  $\pri_{\ell_k,i}$, each agent $i \in C_k$ chooses $a_i=\pri_{\ell_k,i}$.  

Off path, any message that does not come from one's leader  or is not in $\{-1,+1\}$ is ignored. If an agent does not receive a message in $\{-1,+1\}$ from their leader, they  choose an action uniformly at random.\footnote{We do not address subgame perfection in this paper. Nevertheless, we believe that, under appropriate definitions, these equilibria can be made sequential.}

This strategy profile is action-symmetric. It is straightforward to verify that it is an equilibrium. First, a leader has no incentive to deviate: since they are a member of their own community, adhering to the equilibrium guarantees coordination with all neighbors inside that community. Even for a leader with neighbors in another community, sending non-empty messages to neighbors outside the community cannot improve payoffs because such messages are ignored by those neighbors under the prescribed behavior. Moreover, since leaders do not observe the messages other leaders sent to their communities, they cannot hope to coordinate with neighbors from outside their community with probability greater than one half, which they already achieve in this equilibrium. The same holds for non-leaders: following the leader guarantees coordinating with members of their community, while coordinating with members of other communities is impossible to achieve with probability more than one half, since leaders only send messages to their own communities. Finally, non-leaders have no incentive to send non-empty messages, since these are ignored.

Note that implementing a leader equilibrium does not require that the communities $C_k$ be connected. Nevertheless, since our communication model could be a reduced form of a more detailed process in which only direct neighbors communicate, it is nice to have each $C_k$ connected. As we note above, the definition of $(\eps,r)$-amenability is equivalent to one in which we do not require the communities to be connected.

We use leader equilibria to show our first main result.
\begin{theorem}
    \label{thm:amenable}
    Suppose $G$ is $(\eps,r)$-amenable. Then there exists an action-symmetric leader equilibrium with radius of communication $r$ and expected inefficiency at most $\eps$.
\end{theorem}
Thus, if we think of $\eps$ as very small, we have that on graphs that are very amenable, agents can coordinate very well. Note, that the converse of the theorem statement also holds: if $G$ admits a leader equilibrium with inefficiency at most $\eps$, the set of edges $D$ that connects members of distinct communities witnesses the $(\eps,r)$-amenability of $G$. Of course, this result does not preclude the existence of equilibria that are more efficient than leader equilibria.
\begin{proof}[Proof of Theorem~\ref{thm:amenable}]
    Suppose $G=(N,E)$ is $(\eps,r)$-amenable. Then there is a set of edges $D \subseteq E$ such that $|D| \leq \eps|E|$ and each connected component of $(N,E\setminus D)$ is of radius at most $r$. Denote by $\{C_1,C_2,\ldots\}$ the connected components of the graph, which form a partition of $N$. 
    
    Consider a leader equilibrium associated with this partition. In this equilibrium, members of the same community coordinate perfectly, and members of different communities coordinate with probability one half. Since neighboring members of different communities must be connected by an edge in $D$, the total expected inefficiency is at most $|D|$, and thus the average expected inefficiency is at most $|D|/|E| \leq \eps$.

\end{proof}

\section{Impossibility of coordination on non-amenable graphs}

Theorem~\ref{thm:amenable} shows that efficient coordination is possible on an amenable graph. We now turn to the other direction and show that non-amenable graphs do not admit efficient coordination. On such graphs every finite set has a large surface-to-volume ratio. As a result, any attempt to partition the graph into communities necessarily creates many agents on the boundary, who face a high probability of mismatch with their neighbors. This precludes the existence of a leader equilibrium with low inefficiency. Our next result shows that on a non-amenable graph no equilibrium of any kind can achieve low inefficiency. In fact, even if we ignore incentives and consider general strategy profiles, the resulting actions will still generate a large amount of mismatch. Thus, amenability is not only a sufficient, but also a necessary condition for the existence of efficient coordination. The next theorem makes this precise. 

\begin{theorem}
\label{thm:non-amenable}
    Suppose $G$ admits an action-symmetric strategy profile with a radius of communication $r$ and average expected inefficiency $\varepsilon$. Then $G$ is $(\sqrt{8\varepsilon}, r)$-amenable. 
\end{theorem}

This result provides a partial converse to Theorem~\ref{thm:amenable}. It is partial since $\sqrt{8\varepsilon}$  is much larger than $\eps$ (for small $\eps$) and so we cannot preclude that for intermediate ranges there are non-leader equilibria that achieve low inefficiency. Nevertheless, these results together qualitatively characterize amenability as a necessary and sufficient condition for low inefficiency. Since a graph is $(\eps,r)$-amenable if and only if it admits a leader equilibrium with inefficiency $\eps$ and a radius of communication $r$, this theorem immediately implies (by a transformation of $\eps$) that if the best leader equilibrium has inefficiency $\eps$, then no equilibrium can achieve inefficiency less than $\eps^2/8$. In other words, this result shows that if there is a strategy profile that achieves low inefficiency, then there is a leader equilibrium that achieves low inefficiency.

As noted above, incentives do not play a role here: this result applies to any strategy profile. We accordingly rephrase it in a more general, probabilistic framework. The advantage will be that this version is easily applied to a large class of related economic models, as we explain in \S\ref{sec:general-app}.
\begin{theorem}
    \label{thm:non-amenable-general}
    Let $G = (N,E)$ be a finite graph, and let $(Z_i)_{i \in N}$ be independent random variables associated to each vertex. Let $(X_i)_{i \in N}$ also be random variables associated to each vertex, with each $X_i$ measurable with respect to the sigma-algebra generated by $(Z_j)_{j \in B_r(i)}$. Suppose that $\E{X_i}=0$ and $\E{X_i^2}=1$ for all $i$. If 
    \begin{align*}
        \frac{1}{|E|}\sum_{(i,j)\in E}\frac{1}{2}\E{(X_i-X_j)^2} \leq \eps,
    \end{align*}
    then $G$ is $(\sqrt{8\eps},r)$-amenable.
\end{theorem}

Theorem~\ref{thm:non-amenable} is an immediate corollary, as each action $a_i$ (corresponding to $X_i$) is measurable with respect to the sigma-algebra generated by the mixed strategies (corresponding to $Z_j$) of the agents in $B_r(i)$.

To prove Theorem~\ref{thm:non-amenable-general}, we start with random variables $(X_i)_{i \in N}$ satisfying the hypothesis of the theorem. These correspond to an action profile on $G$ with small average expected inefficiency $\varepsilon$.

In the first step, for each agent we construct a probability distribution over the agents in her $r$-neighborhood that measures how much each agent influences her action. Low inefficiency implies that neighboring agents rarely take different actions. We show that this, in turn, forces their influence distributions to be very close to each other in terms of total variation distance.

In the second step, we use these influence distributions to associate to each agent a  leader. Since neighboring agents have similar distributions, the probability that they end up choosing the same leader is high. Agents that select the same leader form a community and choose the same action, so only a small fraction of agents lie on the boundaries between communities. This gives a partition of the graph into small communities with low surface-to-volume ratio, which characterizes amenable graphs. We explain the proof ideas in more detail below.

\subsection{Shapley influence distributions}
Let $Z = (Z_i)_{i \in N}$ be independent (but not necessarily identically distributed) random variables in a standard Borel probability space, indexed by a finite or countable set $N$. For a finite $S \subseteq N$ denote $Z_S = (Z_i)_{i \in S}$.

Let $X$ be a random variable with mean zero and unit variance, measurable with respect to $\sigma(Z)$, i.e., $X = f(Z)$ for some measurable $f$. We would like to calculate an influence distribution: a probability measure $\mu_X$ over $N$ that captures how much each $Z_j$ influences $X$. In our setting, this will measure how much each agent's action was influenced by each other agent in its $r$-neighborhood.

Formally, let $B(Z)$ be the collection of real random variables $X$ that are measurable with respect to the sigma-algebra generated by $Z$ and such that $\E{X}=0$,  $\E{X^2}=1$. An influence distribution map $X \mapsto \mu_X$ assigns to each $X \in B(Z)$ a probability measure $\mu_X$ over $N$. We require $\mu_X$ to be supported on those $j$'s that determine $X$. That is, for any $S \subseteq N$ such that $X$ is measurable with respect to $\sigma(Z_S)$, it holds that $\mu_X(S)=1$.

A desirable property of an influence distribution map is that if $X,Y \in B(Z)$ are close, then the influence distributions $\mu_X$ and $\mu_Y$ should be close. In particular, for our purposes we will need that the total variation distance between $\mu_X$ and $\mu_Y$ should be small whenever $X$ and $Y$ are close in $L^2$. Recall that the total variation distance between two probability measures $\mu,\nu \in \Delta(N)$ is given by
\begin{align*}
    \tv{\mu,\nu} = \sup_{A \subseteq N}|\mu(A)-\nu(A)| = \sum_{i \in N}\frac{1}{2}|\mu(i)-\nu(i)|.
\end{align*}

To this end, given random variable $X \in B(Z)$, define a cooperative game $v_X$ for the set of player $N$ via
\begin{align}
    \label{eq:cooperative-game}
    v_X(S) = \Var\left(\E{X}{Z_S}\right).
\end{align}
That is, $v_X$ assigns to each non-empty finite coalition $S \subseteq N$ the variance of the random variable $\E{X}{Z_S}$. This is a measure for how much the collection of random variables $Z_S$ influences $X$. To see this, it is helpful to take a geometric perspective: $\E{X}{Z_S}$ is the projection, in the Hilbert space of square-integrable random variables, of the random variable $X$ to the subspace of $\sigma(Z_S)$-measurable random variables. The variance is the square of the norm, and when this is high $Z_S$ contains much information about $X$.

Let $\mu_X \colon N \to \R$ be the Shapley values\footnote{When $N$ is infinite the usual definition of the Shapley value cannot be applied, as one cannot draw uniformly a permutation of $N$. Instead, we use the Harsanyi dividend definition of the Shapley value, which is still well-defined.} of the game $v_X$.  Since $v_X \geq 0$, $v_X$ is monotone, and $v_X(N) = \Var(X)=1$, $\mu_X$ is indeed a probability measure on $N$. 

We call $\mu_X$ the Shapley influence distribution of $X$ with respect to $(Z_i)_i$. It is easy to see that if $X$ is measurable with respect to $\sigma(Z_S)$, and if $j \not \in S$, then $v_X(T \cup \{j\}) = v_X(T)$ for all $T$---i.e., $j$ is a null player---and thus $\mu_X(j)=0$. Hence the Shapley influence distribution is indeed an influence distribution.

The next proposition is the main technical contribution of this paper, showing that the Shapley influence distribution map has a Lipschitz property.
\begin{proposition}
  \label{prop:l2l1}
  Let $Z=(Z_j)_{j \in N}$ be independent random variables indexed by a finite or countable set $N$.  Then  for all $X,Y \in B(Z)$ the Shapley influence distribution map $\mu$ satisfies $\tv{\mu_X,\mu_Y} \leq
      \sqrt{\E{(X-Y)^2}}$.
\end{proposition}
Conceptually, this result shows that the Shapley influence distribution map is stable to small perturbations.\footnote{In Appendix~\ref{sec:shapley-stability} we discuss the stability of classical Shapley values, which are not Lipschitz.} This will be useful to us, as we will apply it to $X=a_i$ and $Y=a_j$ for a pair of highly correlated neighbors $(i,j)\in E$, and will be able to conclude that their Shapley influence distributions are close. The remainder of this section is devoted to the proof of this proposition.

Denote by $\L$ the set of $\sigma(Z)$-measurable random variables $X$ such that $\E{X}=0$ and $\E{X^2}<\infty$. The set $\L$ is a separable real Hilbert space, when equipped with the inner product $(X,Y) = \E{X \cdot Y}$. Denote by $\L_i$ the subspace of $\sigma(Z_i)$-measurable elements of $\L$. We observe that $\L_i$ is a closed subspace of $\L$. For each $i$, let $\U_i$ be a (finite or countable) orthonormal basis of $\L_i$. An orthonormal basis of $\L$ is then the collection $\U$ of random variables of the form $U = \prod_{i \in S}U_i$ where $S$ is a finite subset of $N$ and each $U_i$ is in $\U_i$.  We can thus write any $X \in \L$ as an (infinite) linear combination of elements of $\U$:
\begin{align*}
  X = \sum_{U \in \U}\hat x_U U, 
\end{align*}
where $\hat x_{U} = (X,U) \in \R$ is the coefficient of $U \in \U$. By Parseval, we have that for $X,Y \in \L$
\begin{align*}
  (X,Y) = \sum_{U \in \U}\hat x_{U} \cdot \hat y_{U}.
\end{align*}

Given a finite, non-empty $T \subseteq N$, let $\U_T$ be the collection of all such $U$, so that $\U = \cup_T \U_T$, and
\begin{align*}
  X = \sum_T\sum_{U \in \U_T}\hat x_U U.
\end{align*}

Using this, for a finite $S \subseteq N$ we can write
\begin{align*}
\E{X}{Z_S}=\E{\sum_T\sum_{U \in \U_T} \hat x_{U}U}{Z_S}
= \sum_T \sum_{U \in \U_T}\hat x_{U} \E{U}{Z_S}.
\end{align*}
Recall that for $U \in \U_T$,  $U = \prod_{i \in T}U_i$, where each $U_i$ is $\sigma(Z_i)$-measurable, and thus this is a product of independent random variables. Note that
\begin{align*}
  \E{U_i}{Z_S} =
  \begin{cases}
    U_i& \text{if } i\in S\\
    0& \text{if } i\not \in S.
  \end{cases}
\end{align*}
Hence, for $U \in \U_T$,
\begin{align*}
  \E{U}{Z_S} = 
  \begin{cases}
    U&\text{if } T \subseteq S\\
    0&\text{if } T\not \subseteq S.
  \end{cases}
\end{align*}
Thus,
\begin{align*}
  \E{X}{Z_S} = \sum_{T\subseteq S} \sum_{U \in \U_T}\hat x_{U}U.
\end{align*}
Since $\U$ is an orthonormal basis, we get by Parseval that
\begin{align*}
  v_X(S) = \Var\left(\sum_{T\subseteq S}\sum_{U \in \U_T}\hat x_{U}U\right) = \sum_{T\subseteq S}\sum_{U \in \U_T}\hat x_{U}^2.
\end{align*}
Denote
\begin{align}
  \label{eq:dividend}
  \hat x_T^2 = \sum_{U \in \U_T}\hat x_{U}^2,
\end{align}
so that
\begin{align*}
  v_X(S) =  \sum_{T\subseteq S}\hat x_T^2.
\end{align*}
This is exactly the Harsanyi dividend representation of a cooperative game \citep{harsanyi1982simplified}, with dividends $\hat x_T^2$. Therefore, by the standard Shapley value formula in terms of Harsanyi dividends,
\begin{align*}
\varphi_k(v_X)=\sum_{S\ni k}\frac{1}{|S|}\hat x_S^2.
\end{align*}

Fix $X,Y \in \L$ with $\E{X^2}=\E{Y^2}=1$. By the definition of $\mu$,
\begin{align*}
  \tv{\mu_X,\mu_Y} = \frac{1}{2}\sum_i|\mu_X(i)-\mu_Y(i)| =
  \frac{1}{2}\sum_i\left|\sum_{S \ni i}\frac{1}{|S|} \hat x_S^2-\sum_{S \ni
    i}\frac{1}{|S|} \hat y_S^2\right|.
\end{align*}
By the triangle inequality we have
\begin{align*}
  \tv{\mu_X,\mu_Y} = \frac{1}{2}\sum_i\left|\sum_{S \ni
    i}\frac{\hat x_S^2-\hat y_S^2}{|S|}\right| \leq \frac{1}{2}\sum_S|\hat x_S^2-\hat y_S^2|.
\end{align*}
Recalling the definition of $\hat x_S^2$ in \eqref{eq:dividend},
\begin{align*}
  \tv{\mu_X,\mu_Y} \leq \frac{1}{2}\sum_S\left|\sum_{U \in \U_S}\hat x_{U}^2-\sum_{U \in \U_S}\hat y_{U}^2\right|,
\end{align*}
and by another application of the triangle inequality
\begin{align*}
  \tv{\mu_X,\mu_Y} \leq \frac{1}{2}\sum_{U \in \U}\left|\hat x_{U}^2-\hat y_{U}^2\right|.
\end{align*}
By Cauchy-Schwarz,
\begin{align*}
  \sum_{U \in \U}|\hat x_U^2-\hat y_U^2| = \sum_{U \in \U}|\hat x_U-\hat y_U|\cdot |\hat x_U+\hat y_U| \leq
  \sqrt{\sum_{U \in \U}(\hat x_U-\hat y_U)^2}\cdot \sqrt{\sum_{U \in \U}(\hat x_U+\hat y_U)^2}.
\end{align*}
By Parseval,  $\sum_U(\hat x_U-\hat y_U)^2=\E{(X-Y)^2}$  and
$\sum_U(\hat x_U+\hat y_U)^2=\E{(X+Y)^2}$. The latter is at most $4$, and so
\begin{align*}
  \tv{\mu_X,\mu_Y} \leq \sqrt{\E{(X-Y)^2}}.
\end{align*}
This completes the proof of Proposition~\ref{prop:l2l1}.

\begin{remark}
    Proposition~\ref{prop:l2l1} is related to \cite[Lemma 2.9]{galicza2024sparse}, which establishes a looser relationship between covariances and the ``clue'', a way of measuring the dependence of $X$ on the input variables $(Z_j)_{j\in N}$ using Fourier analysis on the hypercube. The proof of Proposition~\ref{prop:l2l1} shows that this ``clue'' is in fact \emph{equal} to the Shapley influence distribution when the $Z_j$ are i.i.d.\ unbiased coinflips, so that our Proposition~\ref{prop:l2l1} strengthens and clarifies their Lemma 2.9. 
\end{remark}

\subsection{Coupling}

In addition to the influence distributions of Proposition~\ref{prop:l2l1}, to prove Theorem~\ref{thm:non-amenable-general} we will make use of what we call grand couplings.

Let $\mu_1,\mu_2$ be probability measures on a finite or countable set $N$. A coupling of these measures is a pair of random variables $L_1,L_2$ defined on the same probability space and such that $L_i$ has a distribution $\mu_i$. We will be interested in couplings that minimize the probability that $L_1\neq L_2$. A classical result \citep[due to Doeblin, see][]{lindvall1991w} is that the minimum of $\Pr{L_1 \neq L_2}$ over all couplings of $\mu_1,\mu_2$ is achieved and is exactly equal to the total variation distance $\tv{\mu_1,\mu_2}$.

One can ask the same question for three measures, $\mu_1,\mu_2,\mu_3 \in \Delta(N)$. In this case it turns out that one can always find a coupling (or perhaps a throupling) $L_1,L_2,L_3$ such that for each $i,j \in \{1,2,3\}$ it holds that $\Pr{L_i \neq L_j} \leq 2 \tv{\mu_i,\mu_j}$. Remarkably, the same holds for coupling of more than three measures. In fact, the following proposition, which is due to \cite{kleinberg2002approximation} \citep[see also][for a closely related idea]{broder1997resemblance}, shows that one can couple all the probability measures on $N$ and achieve the same guarantee. A particularly good exposition is given by \cite{angel2019pairwise}. 

\begin{proposition}[\cite{kleinberg2002approximation,angel2019pairwise}]
\label{prop:grand-coupling}
Let $N$ be a finite or countable set, and let $(\mu_i)_{i \in I}$ be the set of all probability measures on $N$, indexed by a set $I$. Then there exists a probability space with random variables $(L_i)_{i \in I}$ such that $L_i$ has a distribution $\mu_i$, and for each $i,j \in I$ it holds that $\Pr{L_i \neq L_j} \leq 2 \tv{\mu_i,\mu_j}$.
\end{proposition}

We describe a simple construction that achieves this, given by \cite{angel2019pairwise}; we refer the reader to the proof there. Let $(E_k)_{k \in N}$ be i.i.d.\ random variables with the unit exponential distribution (i.e., $\Pr{E_k \leq x}=1-\ee^{-x}$ for $x \geq 0$) defined on a common probability space. For each $\mu_i \in \Delta(N)$, let 
\begin{align*}
    L_i = \argmin_{k \in N}\frac{E_k}{\mu_i(k)}.
\end{align*}
Then $\Pr{L_i=k}=\mu_i(k)$ and 
\begin{align*}
    \Pr{L_i \neq L_j} \leq \frac{2\tv{\mu_i,\mu_j}}{1+\tv{\mu_i,\mu_j}} \leq 2\tv{\mu_i,\mu_j}.
\end{align*}
The relevance of Proposition~\ref{prop:grand-coupling} to amenability of graphs was first noted by \cite{hutchcroft2024relation}. We follow a similar approach, namely, to use the grand coupling to construct random partitions of the vertex set according to common outputs, so that the average surface-to-volume ratio of cells is bounded by twice the average total variation distance between neighbors. 
\subsection{Proof of Theorem~\ref{thm:non-amenable-general}}

Given Propositions~\ref{prop:l2l1} and~\ref{prop:grand-coupling}, we are ready to prove the main result of this section.
\begin{proof}[Proof of Theorem~\ref{thm:non-amenable-general}]

Since each $X_i$ is measurable with respect to the sigma-algebra generated by $(Z_j)_{j \in B_r(i)}$, by Proposition~\ref{prop:l2l1} we can assign to each $X_i$ a probability measure $\mu_i = \mu_{X_i}$ on $N$ such that $\mu_i(B_r(i))=1$ and such that for any pair of vertices $i, j\in N$, the total variation distance $\tv{\mu_i,\mu_j}$ is at most $\sqrt{\E{(X_i-X_j)^2}}$. Using the grand coupling of Proposition~\ref{prop:grand-coupling} we can couple all these measures in a process $(L_i)_{i\in N}$ taking values in $N$ such that $L_i$ has law $\mu_i$ and
  \begin{align*}
    \Pr{L_i \neq L_j} \leq 2\tv{\mu_i,\mu_j} \leq 2\sqrt{\E{(X_i-X_j)^2}}.
  \end{align*}
  
  Let $D$ be the (random) subset of $E$ given by
  \begin{align*}
    D = \{(i,j) \in E\,:\, L_i \neq L_j\}.
  \end{align*}
  Then $L_i$ is constant on each connected component of $G'=(N,E
  \setminus D)$. Since $\mu_i(B_r(i))=1$ it follows that $L_i \in B_r(i)$, and so the connected component of $i$ is a subset of $B_{r}(L_i)$.

  Finally,
  \begin{align*}
    \E{|D|} = \sum_{(i,j) \in E}\Pr{L_i \neq L_j} \leq 2\sum_{(i,j)
      \in E}\sqrt{\E{(X_i-X_j)^2}} \leq  \sqrt{8\eps}|E|,
  \end{align*}
  where the last inequality follows from Jensen's inequality.
\end{proof}

\begin{remark}
\label{rem:virag}
Both the results and proof of Theorem~\ref{thm:non-amenable-general} is reminiscent of the methods of \cite*{csoka2020entropy}, who worked under a ``uniform for small sets'' notion of non-amenability rather than our weaker ``statistical'' notion of non-amenability (i.e., non-hyperfiniteness). The authors of that paper work primarily with the entropy rather than the variance, but discuss variance versions of their inequalities in Section 4.2. Their method is somewhat different than ours and does not identify the Lipschitz property of Proposition~\ref{prop:l2l1}.
\end{remark}

\section{More on coordination on amenable graphs}

Section \ref{sec:coordination_private} gives a benchmark construction for achieving low-inefficiency coordination on amenable graphs through leader equilibria. We now ask how far this conclusion extends when coordination is more constrained. 
We first weaken the communication structure: agents cannot rely on private communication and instead use (locally) public messages. We then turn to the objects that the benchmark construction takes as given---the partition and the leaders---and show how agents can coordinate on both through local communication. In both cases, the main conclusion survives: amenability still allows agents to coordinate with low inefficiency. Nevertheless, the restrictions on communication and local implementation come at some cost in efficiency.

\subsection{Equilibria with locally public communication}

Theorem \ref{thm:amenable} shows that on amenable graphs there are leader equilibria that achieve low inefficiency. These equilibria make critical use of private communication and, in particular,
hinge on the leader being able to communicate selectively with the members of their community. In light of this, one may wonder what can be achieved without private communication, using only public messages to coordinate.

In this section, we consider the following variant of the communication stage. Instead of sending private messages, each agent $i$ can choose a a single message $m_i$ to broadcast to all agents in $B_r(i)$. We call this message  \emph{locally public}, as the same message is observed by all agents in $i$'s $r$-neighborhood. This can capture a public action, such as posting a sign
on one's yard. %As with private messages, locally public messages sent directly to agents within distance $r$ can succinctly model a more complicated process.

Formally, agents first choose locally public messages simultaneously. After each agent $i$ observes $m_{B_r(i)}=(m_j)_{j\in B_r(i)}$, agents simultaneously choose actions. Strategies, mixed strategies, and action-symmetry are defined as before, with locally public messages
replacing private messages. Thus a pure strategy of agent $i$ is a pair
$\sigma_i=(m_i,o_i)$, where $o_i$ assigns an action to every realization of
$m_{B_r(i)}$.

The difficulty is that the leader's message is now observed by everyone in their $r$-neighborhood, not only by the members of their community. Indeed, consider the strategy profile in which each leader $\ell_k$ chooses uniformly at random a locally public message $m_{\ell_k}\in\{-1,+1\}$, and all agents $i\in C_k$ take action $a_i=m_{\ell_k}$. Suppose there is an agent $i$ in a community $C$ who has more neighbors in a different community $C'$ than in $C$, and who is also within the radius of communication of $\ell'$, the leader of $C'$. Then $i$ would have a profitable
deviation to choosing $a_i=m_{\ell'}$. Since $(\varepsilon,r)$-amenability alone does not exclude such an agent, the private-message construction from Theorem \ref{thm:amenable} cannot be implemented directly with locally public messages.

Our main result of this section is that even with only locally public communication, there still exist low-inefficiency leader equilibria on amenable graphs, although with slightly worse guarantees.

\begin{theorem}
    \label{thm:amenable-no-comm}
    Suppose $G$ is $(\eps,r)$-amenable. Then there exists a public communication, action-symmetric leader equilibrium with  radius of communication $r$ and expected inefficiency at most $\eps \dmax$.
\end{theorem}
The main idea behind this proof is the notion of stable communities. We say that a community $C \subseteq N$ is \textit{stable} if for all $i \in C$ it holds that $|N_i \cap C| \geq |N_i \setminus C|$. That is, each agent in $C$ has at least as many neighbors in $C$ as outside of $C$. The advantage of a stable community is that members do not have an incentive to deviate from the community's coordinated action, even if they know what their neighbors from outside the community will do.\footnote{A stable set is a $1/2$-cohesive set in the terminology of \cite{morris2000contagion}.}

Suppose that $(C_1,C_2,\ldots,C_K)$ are disjoint stable communities, and let $U = N \setminus \cup_k C_k$ be the agents who are not members of any of these communities, so that $(C_1,C_2,\ldots,C_K,U)$ is a partition of $N$. For each $k$, fix a designated leader $\ell_k$ for $C_k$. We require that $C_k$ is contained in $B_r(\ell_k)$, and that $\ell_k$ is either a member of $C_k$ or an agent that does not belong to any community. %If $\ell_k$ is not a member of $C_k$, we call him a banished leader. 
The leader $\ell_k$ of each stable community will send a locally public message $m_{\ell_k}$ chosen uniformly from $\{-1,+1\}$. Non-leaders will send an empty public message. Members of stable communities will follow their leader's message. We do not specify explicitly what actions members of $U$ choose, and instead show in the proof that there exist action-symmetric strategies for these agents that make this an equilibrium. Its inefficiency will be at most the size of the set of edges
\begin{align*}
    D =  \left(\cup_k \partial C_k\right) \bigcup \{(i,j)\,:\, i \in U, j \in N_i\},
\end{align*}
since miscoordination can only occur on community boundaries, or outside the communities on edges adjacent to $U$. The formal proof is given in the Appendix \ref{app:amanble-no-comm}.

Given this, to prove Theorem~\ref{thm:amenable-no-comm}, it remains to show that if $G$ is $(\eps,r)$-amenable, then it admits a  set of edges $D$ of size at most $\eps\dmax|E|$ whose removal leaves stable communities and a set of isolated vertices, such that each community $C_k$ has a leader $\ell_k$ whose distance to every member of $C_k$ is at most $r$.
\begin{proposition}
\label{prop:stable}
     Suppose that a finite graph $G=(N,E)$ is $(\eps,r)$-amenable. Then there exists a set of edges $D \subseteq E$ such that $|D| \leq \eps \dmax|E|$, and every connected component of $G'=(N,E \setminus D)$is either a singleton or a stable set $C$. Moreover, if
$C_1,\ldots,C_K$ are the stable components of $G'$, then there exist distinct
vertices $\ell_1,\ldots,\ell_K$ such that, for each $k$, $\ell_k$ is either in
$C_k$ or is a singleton component of $G'$, and $C_k\subseteq B_r(\ell_k)$. 
\end{proposition}
This proposition states that we can remove a small number of edges and be left with stable communities, and some isolated vertices. Each stable community has a leader; the leader need not be a member of the community, but if they are not a member then they are one of the isolated vertices, and are still within distance $r$ to each member of the community. Distinct communities have distinct leaders. 

The idea behind the proof of this proposition is simple: start with communities separated by a small set of edges, whose existence is given by the fact that $G$ is $(\eps,r)$-amenable. If a community is not stable, it has a member with more neighbors outside the community than inside it. Banish this member to its own singleton community, and repeat until all communities are stable or singletons. The crux is to show that this ends after a small number of steps.
\begin{proof}[Proof of Proposition~\ref{prop:stable}]

Let $(C_1,\dots,C_K)$ witness $(\varepsilon,r)$-amenability, so each $C_k$ has radius at most $r$ and
$\sum_k |\partial C_k|\le \varepsilon |E|$. If $C_k$ is stable, set $C_k':=C_k$. Otherwise, there exists
$i\in C_k$ such that $|N_i\cap C_k|<|N_i\setminus C_k|$. Let $S:=C_k\setminus\{i\}$. Removing $i$ decreases
the boundary by at least $1$: it removes $|N_i\setminus C_k|$ boundary edges and creates $|N_i\cap C_k|$ new
boundary edges, so $|\partial S|\le |\partial C_k|-1$. If $S$ is still unstable, repeat the same operation. Since
$|\partial(\cdot)|$ is a nonnegative integer that drops by at least $1$ each time, this procedure terminates after
at most $|\partial C_k|$ deletions, yielding a set $C_k'\subseteq C_k$ that is either empty or stable.

Let $U:=N\setminus \bigcup_k C_k'$ be the set of vertices not contained in the resulting stable sets. By the bound
above, the peeling procedure deletes at most $|\partial C_k|$ vertices from $C_k$, hence
$|U|\le \sum_k |\partial C_k|$. Now let $D$ be the union of $\bigcup_k \partial C'_k$ and all edges incident to the vertices in $U$. Then, 
$$ |D| \le \sum_k |\partial C'_k| + d_{\max}|U|.$$

Since the number of boundary edges decreased by at least the size of $U$,
$$\sum_k|\partial C'_k| + d_{\max}|U| \leq  (\sum_k |\partial C_k| -|U|)+ d_{\max}|U| \leq d_{\max} \sum_k |\partial C_k|\le \varepsilon d_{\max}|E|.$$

In $G'=(N,E\setminus D)$, every vertex in $U$ is isolated, hence a singleton component. Any other vertex lies in some
$C_k'\subseteq C_k$, which is stable. Moreover, since $C_k$ has radius at most $r$, there exists $\ell_k\in C_k$ within distance at most $r$ from all $i\in C_k$, and therefore also from all $i\in C'_k$. These leaders are distinct, because each $\ell_k$ is chosen from
the original part $C_k$, and the parts $(C_1,\ldots,C_K)$ are disjoint. Note that $\ell_k$ need not belong to $C'_k$, as it may be among the deleted vertices of $C_k$. 
\end{proof}

\subsection{Local strategy profiles}
The proof of Theorem~\ref{thm:amenable} relies on a given partition of the social network into connected communities with small boundaries. One may object that agreeing on this partition requires some global coordination, and that it would be interesting to know whether coordination on actions can be achieved without first coordinating on communities. 

We show that this is possible by allowing agents to coordinate locally through both public and private communication. Thus, we again alter the communication stage and allow agents to simultaneously send both a private message and a public message. Agents will use the public messages to construct communities, and to choose their leaders. As we show in the Appendix \ref{app:network-symmetric}, the choice of leaders is a relatively easy problem to solve: given a partition, we can choose a leader locally at random using public messages. Agents then proceed with the leader equilibrium, coordinating on actions using private messages. The more challenging problem is the local choice of partition.

\begin{comment}

Leader equilibria also require a choice of leader, which again can be viewed as a coordination problem. As we show in the appendix, this is a relatively easy problem to solve: given a partition, we can choose a leader locally at random using public messages, and then proceed with the leader equilibrium, coordinating on actions using private messages. The more challenging problem is the local choice of partition.
\end{comment}

We study this problem for graphs that ``look the same'' around each vertex, up to radius $r$. An illustrative example is the $n\times n$ torus: a two dimensional grid, in which nodes at the boundaries are connected to the corresponding nodes on the opposite side. In this graph the $r$-neighborhoods of all agents are identical.

We consider strategy profiles that are local, in the sense that they do not require a globally-coordinated partition of the graph into communities, and also do not require a priori coordination on leaders. In this setting, we show that amenability remains a sufficient condition for coordination, but with weaker guarantees. Specifically, we show that for any graph that is $(\eps,r)$-amenable there exists a local, action-symmetric equilibrium with expected inefficiency at most $\varepsilon\left(1+\log\frac{1}{\eps} \right)$. Thus, inefficiency is still small when $\eps$ is small, but larger than what we can guarantee without the locality restriction.

The idea behind the local construction of communities is that agents use public communication to decide locally which edges of the original graph to remove; the communities are then the connected components of the remaining graph. To this end, fix a connected subgraph $F$ of radius $r$; on the torus we take $F$ to be an $(r+1)\times(r+1)$ square, whose surface-to-volume ratio is $\varepsilon = 4/(r+1)$. In general, we will want $F$ to have the lowest surface-to-volume ratio possible among all subgraphs of radius $r$.

Let $F'$ denote a subgraph that is isomorphic to $F$; on the torus, these subgraphs are precisely the translated copies of the square. Using public communication, the agents will decide whether or not to select each $F'$. This decision will be made independently for each such $F'$, including overlapping ones. The probability of selecting $F'$ will be chosen so that each agent will, with high probability, be a member of some selected $F'$. The edges removed are the boundary edges of all the selected $F'$s, with the selection probability being low enough so that the total number of edges removed is unlikely to be too large. The edges adjacent to agents who are not members of any selected $F'$ are also removed, making them singleton communities. The connected components of the remaining graph form the communities, and finally the public messages are also used to select the leader of each community. Figure~\ref{fig:local-torus} illustrates this rule.

\begin{figure}[h]
    \centering
\tikzset{every picture/.style={line width=0.75pt}} %set default line width to 0.75pt        

\begin{tikzpicture}[x=0.75pt,y=0.75pt,yscale=-1.2,xscale=1.2]
%uncomment if require: \path (0,444); %set diagram left start at 0, and has height of 444

%Shape: Rectangle [id:dp1874166500415836] 
\draw  [line width=0.75]  (101,100.02) -- (272,100.02) -- (272,239.52) -- (101,239.52) -- cycle ;
%Shape: Rectangle [id:dp327806334514835] 
\draw  [color={rgb, 255:red, 255; green, 0; blue, 0 }  ,draw opacity=1 ][line width=1.5]  (150,121.52) -- (231,121.52) -- (231,199.52) -- (150,199.52) -- cycle ;
%Shape: Rectangle [id:dp6282177178240108] 
\draw  [color={rgb, 255:red, 255; green, 0; blue, 0 }  ,draw opacity=1 ][line width=1.5]  (191,161.52) -- (272,161.52) -- (272,239.52) -- (191,239.52) -- cycle ;
%Straight Lines [id:da20479426522213073] 
\draw [color={rgb, 255:red, 255; green, 0; blue, 0 }  ,draw opacity=1 ][line width=1.5]    (140,100) -- (140,158.52) ;
%Straight Lines [id:da9228798956474163] 
\draw [color={rgb, 255:red, 255; green, 0; blue, 0 }  ,draw opacity=1 ][line width=1.5]    (141,158.52) -- (101,158.52) ;
%Straight Lines [id:da03614790399724499] 
\draw [color={rgb, 255:red, 255; green, 0; blue, 0 }  ,draw opacity=1 ][line width=1.5]    (231,100) -- (231,141.52) ;
%Straight Lines [id:da24899444430233852] 
\draw [color={rgb, 255:red, 255; green, 0; blue, 0 }  ,draw opacity=1 ][line width=1.5]    (272,141.52) -- (231,141.52) ;
%Straight Lines [id:da9000774110239854] 
\draw [color={rgb, 255:red, 255; green, 0; blue, 0 }  ,draw opacity=1 ][line width=1.5]    (170,171.52) -- (101,171.52) ;
%Straight Lines [id:da512331790592515] 
\draw [color={rgb, 255:red, 255; green, 0; blue, 0 }  ,draw opacity=1 ][line width=1.5]    (170,170.52) -- (170,239.52) ;
%Shape: Rectangle [id:dp9681636166451483] 
\draw  [draw opacity=0][fill={rgb, 255:red, 155; green, 155; blue, 155 }  ,fill opacity=0.15 ] (140,100) -- (231,100) -- (231,121.52) -- (140,121.52) -- cycle ;
%Shape: Rectangle [id:dp19714604703320915] 
\draw  [draw opacity=0][fill={rgb, 255:red, 155; green, 155; blue, 155 }  ,fill opacity=0.15 ] (140,121.52) -- (150,121.52) -- (150,170.52) -- (140,170.52) -- cycle ;
%Shape: Rectangle [id:dp8578930884095806] 
\draw  [draw opacity=0][fill={rgb, 255:red, 155; green, 155; blue, 155 }  ,fill opacity=0.15 ] (101,159.52) -- (140,159.52) -- (140,171.52) -- (101,171.52) -- cycle ;
%Shape: Rectangle [id:dp7257169175472128] 
\draw  [draw opacity=0][fill={rgb, 255:red, 155; green, 155; blue, 155 }  ,fill opacity=0.15 ] (170,200) -- (190,200) -- (190,239.52) -- (170,239.52) -- cycle ;
%Shape: Rectangle [id:dp41920843160711063] 
\draw  [draw opacity=0][fill={rgb, 255:red, 155; green, 155; blue, 155 }  ,fill opacity=0.15 ] (231,141.52) -- (272,141.52) -- (272,161.52) -- (231,161.52) -- cycle ;

% Text Node
\draw (112,121.4) node [anchor=north west][inner sep=0.75pt]  [font=\small]  {$C_{1}$};
% Text Node
\draw (177,137.4) node [anchor=north west][inner sep=0.75pt]  [font=\small]  {$C_{2}$};
% Text Node
\draw (242,114.4) node [anchor=north west][inner sep=0.75pt]  [font=\small]  {$C_{3}$};
% Text Node
\draw (115,211.4) node [anchor=north west][inner sep=0.75pt]  [font=\small]  {$C_{6}$};
% Text Node
\draw (151,178.92) node [anchor=north west][inner sep=0.75pt]  [font=\small]  {$C_{4}$};
% Text Node
\draw (203,172.4) node [anchor=north west][inner sep=0.75pt]  [font=\small]  {$C_{5}$};
% Text Node
\draw (241,211.4) node [anchor=north west][inner sep=0.75pt]  [font=\small]  {$C_{7}$};

\end{tikzpicture}

    \caption{In the equilibrium constructed in the proof of Theorem~\ref{thm:amenable-transitive}, agents choose at random which tiles should be selected. Communities are the intersections of the selected tiles. Agents outside the selected tiles form singleton communities. }
    \label{fig:local-torus}
\end{figure}
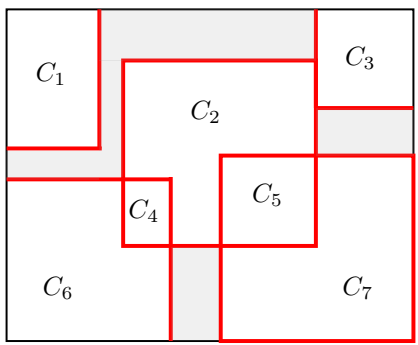

Note that since agents prefer not to be at the boundaries of communities, care needs to be taken to ensure that the local choice of communities is incentivized. Our main tool is the computer-science technique of ``secret-sharing'' \citep{shamir1979share}: when a group of agents $F'$ needs to choose at random with probability $p$ whether to select their subgraph, we have each send a public message $m_i$ chosen uniformly from $[0,1]$, add these messages, and select $F'$ if the fractional part of their sum is in $[0,p]$. It is easy to verify that as long as $F'$ is not a singleton, agents are indifferent between sending any message in $[0,1]$. Since we have to repeat this for many subgraphs $F'$, the public message $m_i$ will in fact have to be a tuple of many independent messages, each distributed uniformly on $[0,1]$.

We formalize the notion of graphs that ``look the same'' around each vertex as follows. We say that $G=(N,E)$ is $r$-\emph{locally-transitive} if the $r$-neighborhoods of any two vertices are isomorphic as rooted graphs. 
A strategy profile $\sigma$ is network-symmetric (under radius of communication $r$) if any such local isomorphism maps the strategy of one agent to the strategy of the other.   This notion captures the idea that all agents follow the same strategy, rather than having different roles in the strategy profile, avoiding the conceptual need for a central planner.\footnote{Formally, $G=(N,E)$ is $r$-locally-transitive if for any two vertices $i,j \in N$ there exists a bijection $\varphi: B_{r}(i) \to B_{r}(j)$ such that (1) $\varphi (i) = j$, and (2) the bijection preserves vertex adjacency: $(i,j)\in E \iff (\varphi(i), \varphi(j))\in E$. A strategy profile $\sigma$ is network-symmetric if for all $i,j \in N$ and every rooted graph isomorphism $\varphi: B_{r}(i) \to B_{r}(j)$ it holds that $\varphi(\sigma_i) =\sigma_j$, where we extend $\varphi$ to a map $\varphi \colon A_i \to A_j$  in the obvious way.}

\begin{theorem}
    \label{thm:amenable-transitive}
    Suppose $G$ is $(\eps,r)$-amenable and $(2r+1)$-locally transitive. Then there exists a network-symmetric, action-symmetric leader equilibrium with radius of communication $4r$, using both public and private communication, and having  and expected inefficiency at most $\varepsilon\left(1+\log\frac{1}{\eps} \right)$.
\end{theorem}
This theorem shows that efficient coordination is possible on an $(\eps,r)$-amenable graph even without reliance on a global partition.  Nevertheless, in our construction communication is less efficient, by a factor of $1+\log\frac{1}{\eps}$, and requires a larger radius of communication.

\section{Local correlated equilibria and the near-optimality of leader equilibria on cycle graphs}
\label{sec:cycles}

In this section we explore \emph{local correlated equilibria}, an alternative solution concept to coordination games on networks. As the name suggests, these equilibria are correlated equilibria, but where coordination is restricted to not exceed beyond local neighborhoods. This captures a similar spirit to our local messaging assumption, and indeed the Nash equilibria of the messaging game we have discussed so far will be special cases.

Let $N_1,N_2 \subset N$ be non-empty subsets of agents. We say they are $2r$-separated if the graph distance between every $i \in N_1$ and $j \in N_2$ is strictly greater than $2r$. Let $\mu$, a probability measure on $\{-1,+1\}^N$, be the joint distribution of the random variables $(a_i)_{i \in N}$. We say that $\mu$ is $r$-local if for every $N_1,N_2 \subset N$ that are $2r$-separated it holds that the random variables $(a_i)_{i \in N_1}$ and $(a_j)_{j \in N_2}$ are independent.\footnote{In the mathematics literature this is sometimes referred to as $2r$-dependent \citep{burton1993on, holroyd2016finitely}.} We say that $\mu$ is an $r$-local correlated equilibrium if, in addition, it is a correlated equilibrium. Action-symmetry is simple to define for $r$-local distributions: it merely requires that $\mu(a) = \mu(-a)$.

It is immediate that for any strategy profile with radius of communication $r$ in our messaging game that results in actions $(a_i)_{i \in N}$, the joint distribution of these actions is $r$-local. Likewise, if this strategy profile is an equilibrium then this joint distribution is an $r$-local correlated equilibrium. An interesting fact is that there are graphs  in which not every $r$-local $\mu$ arises in this way \citep[see, e.g.][]{burton1993on, holroyd2016finitely}.

In many ways, we find this solution concept more appealing than Nash equilibria of the messaging game. The downside is that we are generally unable to prove an analogue of Theorems~\ref{thm:non-amenable} and~\ref{thm:non-amenable-general} for this solution concept, although we do conjecture that results of this flavor should hold. Nevertheless, the particular case of cycle graphs turns out to be tractable.

Beyond cycle graphs---indeed, on any graph---Theorem~\ref{thm:amenable} shows that on an $(\eps,r)$-amenable graph there is an equilibrium with radius of communication $r$ that achieves inefficiency $\eps$. A natural question is: are there equilibria that are more efficient than leader equilibria? In general, we do not know the answer to this question. Theorem~\ref{thm:non-amenable} yields a lower bound, showing that if the best leader equilibrium has inefficiency $\eps$, then no strategy profile can achieve inefficiency less than $\eps^2/8$. We show that for the particularly simple example of cycle graphs leader equilibria are optimal. In fact, we show that no action-symmetric $r$-local strategy profile can achieve lower inefficiency.

In a cycle graph, the set of agents is $N=\{0,1,\ldots,n-1\}$, and $i$ is connected to $j$ if $i = j \pm 1 \mod n$. We assume that $n$ is much larger than the radius of communication ($n \geq 5r$ suffices), and, moreover, for simplicity, that the number of agents $n$ is a multiple of $2r+1$, so that we can divide the circle into segments of length $2r+1$ (and hence radius $r$) as in Figure~\ref{fig:Z-1d}. The average inefficiency of the corresponding leader equilibrium is $1/(2r+1)$, as we explain in \S~\ref{sec:amenable} above. 

Consider an action-symmetric, $r$-local $\mu \in \Delta(\{-1,+1\})^n$, and suppose that it is the distribution of  $(a_i)_{i \in N}$. These random variables take values in $\{-1,+1\}$, and all have expectation zero, by action-symmetry. Thus, they are vectors in $L_0^1$, the space of zero-mean random variables. This is a vector space, and indeed a Banach space under the metric  $d(X,Y) = \E{|X-Y|}$. Under this metric, each $a_i$ has unit norm.

Since the average inefficiency is $x$, there must be a sequence of agents $i,i+1,\ldots,i+2r+1 \mod n$ with average inefficiency $\frac{1}{2r+1}\sum_{j=i}^{i+2r}\E{|a_j-a_{j+1}|} \leq x$. Hence  $\sum_{j=i}^{i+2r}d(a_j,a_{j+1}) \leq (2r+1)x$. From the triangle inequality if follows that $d(a_i,a_{i+2r+1}) \leq (2r+1)x$. Since $\mu$ is $r$-local, $a_i$ and $a_{i+2r+1}$ are independent, and so $d(a_i,a_{i+2r+1})=1$. We have thus shown that $x \geq 1/(2r+1)$, and so no $r$-local $\mu$ can achieve inefficiency lower than the leader equilibrium.

\section{Conclusion}

\label{sec:general-app}

In this paper, we identify amenability as a geometric condition that captures the possibility of local coordination on a social network. Theorems~\ref{thm:amenable}, \ref{thm:amenable-no-comm} and \ref{thm:amenable-transitive} show that on amenable graphs it is possible to achieve low inefficiency, while Theorem~\ref{thm:non-amenable} shows that low inefficiency implies that a graph is amenable.

Theorem~\ref{thm:non-amenable-general} is a general result stating that if a finite graph admits local random variables that are highly correlated, then it must be amenable. Given the generality of this result, we expect it to have further applications. For example, if we consider a repeated game version of our setting, then agents' actions at time $t$ would depend on their neighbors in a ball of radius $t$, and so on non-amenable graphs high correlation is impossible. This still holds in even more general settings that go beyond pure coordination games. For example, agents could have idiosyncratic preferences over the actions, i.e., independent private types that influence their payoffs, making them prefer to coordinate on a particular action. 

A question that we leave unanswered is whether leader equilibria are always optimal. Equivalently, whether $(\eps,r)$-amenability coincides with having $\eps$ inefficiency in equilibria with a radius of communication $r$. For cycle graphs, we show that this holds in \S\ref{sec:cycles}, but the technique used there does not easily generalize. A possible avenue for tackling this problem could be through applying Shapley values to the $L^1$ norm, rather than the $L^2$ norm that we use. That is, letting the influence distribution be given by the Shapley values of the game $v'_X(S) = \Vert\E{X}{Z_S}\Vert_1$, rather than our current definition \eqref{eq:cooperative-game} of $v_X(S) = \Vert\E{X}{Z_S}\Vert_2^2$. For these Shapley values we do not know how to prove a Lipschitz result corresponding to Proposition~\ref{prop:l2l1}.

\bibliography{refs}

\appendix

\section{Stability of the Shapley value}
\label{sec:shapley-stability}

Proposition~\ref{prop:l2l1} shows a Lipschitz property of the Shapley influence distribution map. This motivates the question of the Lipschitz properties of the Shapley value.

Let $N$ be a finite set of players, and $\R^{2^N\setminus\emptyset}$ be the collection of cooperative games, equivalently, this is the collection $v\colon 2^N \to \mathbb{R}$ with $v(\emptyset) = 0$. The Shapley value is a map $\varphi \colon \R^{2^N\setminus\emptyset} \to \R^N$. Lipschitz properties of the Shapley value capture how stable it is to changes (or measurement errors) in the game. Let $v,w\in \R^{2^N\setminus\emptyset}$, and suppose $\vert v - w \vert_\infty = \max_S |v(S)-w(S)|$ is small. How different are the Shapley values of $v$ and $w$?

There can be a number of ways to quantify the difference between the Shapley values. One is via the sup-norm: $$|\varphi(v)-\varphi(w)|_{\infty} = \max_i|\varphi_i(v)-\varphi_i(w)|.$$ Another is by the total variation norm 
\begin{align*}
    \tv{\varphi(w),\varphi(v)} = \frac{1}{2}|\varphi(w)-\varphi(v)|_1.
\end{align*}
When the former is small, the error in the value of each player is small. When the latter is small, the error in each group of players is small, which is a stronger notion. Proposition~\ref{prop:l2l1} yields the latter type estimate, although in a restricted setting that does not apply to all games. 

Since $\varphi$ is a linear map, understanding the deviations around a game $v$ is the same as understanding the deviations around zero. We thus define the operator norms
\begin{align*}
    \Vert \varphi \Vert_{\infty\to \infty} = \max\{|\varphi(v)|_\infty\,:\, |v|_\infty \leq 1 \}
\end{align*}
and
\begin{align*}
    \Vert \varphi \Vert_{\infty\to 1} = \max\{|\varphi(v)|_1\,:\, |v|_\infty \leq 1 \}.
\end{align*}

It is easy to see that $\Vert \varphi \Vert_{\infty\to \infty}\leq 2$; this follows immediately from the the definition of the Shapley value. One can interpret this fact as a demonstration of the stability of the Shapley value in the sup-norm sense: if we change $v$ by at most $\eps$ in each coordinate, the Shapley value of any player will not change by more than $2\eps$. In contrast, \cite{kumabe_et_al:LIPIcs.ICALP.2024.102}  show that the Shapley value can be far from stable in the total-variation sense: a small change in the game can cause a group of players to have a very large change in their total value---when the number of players is large. Specifically, they show that there's a universal constant $C>0$ such that $\Vert \varphi \Vert_{\infty\to 1}\geq C \log |N|$, and in particular $\Vert \varphi \Vert_{\infty\to 1}$ goes to infinity as the set of players becomes larger. This result also contrasts Proposition~\ref{prop:l2l1}, in which the bound is independent of $|N|$.

\section{Proof of Theorem~\ref{thm:amenable-no-comm}}
\label{app:amanble-no-comm}

Since $G=(N,E)$ is $(\eps,r)$-amenable, by Proposition~\ref{prop:stable} there exists a set of edges $D \subseteq E$ such that $|D| \leq \eps\dmax|E|$, $N=C\cup U$, where $C$ is the union of disjoint stable sets $(C_1,C_2,\ldots)$, and every $i\in U$ is a singleton vertex in the graph $(N,E \setminus D)$. Assign to each community $C_k$ a leader $\ell_k\in C_k\cup U$ such that $C_k \subseteq B_r(\ell_k)$. Denote the set of leaders by $L$. If $\ell_k\not\in C_k$, we call $\ell_k$ a banished leader. 

Let $C = \cup_k C_k$ so that $N$ is the disjoint union of $C$, the agents who are members of communities, and $U$, the agents who are not. We will construct our action-symmetric equilibrium by fixing the behavior of the agents in $C$, and relying on an existence result to extend the strategy profile to the agents in $U$. For both groups, we assume that all private messages are empty, and any private messages are ignored, so that there is no incentive to send private messages.

We start by specifying the behavior of the agents in $C$. For each agent $\in N$, let $L_i$ denote the set of leaders within radius $r$, i.e. $L_i = L \cap B_r(i)$. Our off-path rule specifies that if any agent observes a leader sending a public message outside $\{-1, +1\}$, then agents randomize uniformly. This deters leaders from sending off-path messages. 

The strategy of each leader $\ell_k\in C_{k}$ is as follows:
\begin{enumerate}
    \item choose a public message $\pub_{\ell_k}$ uniformly at random from $\{-1,+1\}$;
    \item take action $a_{\ell_k}=\pub_{\ell_k}$ if $\pub_{\ell}\in \{-1, +1\}$ for every  $\ell \in L_{\ell_k}$;

    \item in the (off-path) event that $\pub_{\ell}$ is not in $\{-1,+1\}$ for some  $\ell \in L_{l_k}$, choose the action $a_{\ell_k}$ uniformly at random from $\{-1,+1\}$.
\end{enumerate}
Fix the strategy of any non-leader  $i \in C_k$ to
\begin{enumerate}
    \item send an empty public message;
    \item take action $a_{i}=\pub_{\ell_k}$ if $\pub_{\ell}\in \{-1, +1\}$ for every  $\ell \in L_{i}$;
    \item in the (off-path) event that $\pub_{\ell}$ is not in $\{-1,+1\}$ for some $\ell \in L_i$, choose the action $a_{i}$ uniformly at random from $\{-1,+1\}$.
\end{enumerate}

Thus, members of communities follow a leader equilibrium using the leader's public messages (as opposed to the private ones used in Theorem~\ref{thm:amenable}), and ignore any off-path messages. Note that these strategies are action-symmetric.

Given that the agents in $C$ follow these strategies, it is not immediately obvious what the best responses are of the agents in $U$. The issue is that depending on the graph structure and the set $D$, it could be that an agent in $U$  receives the public messages of many leaders, and that their neighbors, who are also in $U$, likewise see the same messages. In this case, they could, for example, coordinate to follow the message of one particular leader, or alternatively follow another one. To address this problem, we do not construct their strategies explicitly, but instead apply an existence result.

We define a finite normal form game played by the agents in $U$ at the action stage, conditional on the realized public messages of leaders.  A pure strategy of $i$ is a map $s_i \colon \{-1,+1\}^{L_i} \to \{-1,+1\}$, interpreted as the action $i$ will take in the original game when observing a given tuple of public messages from the leaders in their radius of communication. Given a pure strategy profile $s$, the expected utility of the agents is the expectation of \eqref{eq:one-shot}, their utility in the original game, where the agents in $C$ are dummy players playing as described above, and the agents in $U$ use their strategies in $s$ to choose their actions as a function of the public messages sent by the leaders.

The map $\iota$ defining action-symmetry is what \cite{Nash1951} calls an \emph{automorphism} or \emph{symmetry} of this game, as it preserves payoffs: $u_i(\iota(s_1),\ldots,\iota(s_n)) = u_i(s_1,\ldots,s_n)$. He shows (Theorem 2, p.\ 289) that given an automorphism, there exists a mixed equilibrium that is invariant to the automorphism, which in our setting precisely translates to action-symmetry. Thus, we have an action-symmetric equilibrium $s^*$ of this restricted game of the agents in $U$. 

The agents in $U$ are either banished leaders or non-leaders. For each leader $\ell \in U$ we define a strategy in the original game as follows:
\begin{enumerate}
    \item choose a public message $\pub_{\ell}$ uniformly at random from $\{-1,+1\}$;
    \item if the public messages of the leaders in $L_{\ell}$ are all in $\{-1,+1\}$, take the action specified by applying $s^*_{\ell}$ to these messages (note that $\ell \in L_\ell)$;
    \item in the (off-path) event that at least one leader in $L_\ell$ sends a public message not in $\{-1,+1\}$, choose an action uniformly at random. 
\end{enumerate}

For each non-leader $i \in U$ we define a strategy in the original game as follows:
\begin{enumerate}
    \item send an empty public message;
    \item if the public messages of the leaders in $L_i$ are all in $\{-1,+1\}$, take the action specified by applying $s^*_i$ to these messages;
    \item in the (off-path) event that at least one leader in $L_i$ sends a public message not in $\{-1,+1\}$, choose an action uniformly at random. 
\end{enumerate}

We first show that agents in $C$ are best responding. Non-leaders are best responding because of the stability property of the communities: they have at least as many neighbors in their own community, and these will take the action recommended by the leader, and so no improvement is possible. Deviations other than not following the leader---i.e., sending non-empty public or private messages---are easily ruled out, as all such messages are ignored.

To verify that leaders belonging to the communities are best responding, we need to check that they have no incentive to deviate from choosing their public messages uniformly from $\{-1,+1\}$ and instead choosing (say) $+1$. All other deviations (including sending a public message not in $\{-1,+1\}$ and taking an action other than the message) are again easy to rule out.

Fix a leader $\ell \in C$. If $\ell$ has no neighbors outside its community then clearly there is no profitable deviation, since $\ell$ is guaranteed to coordinate with all its neighbors. Otherwise, let $j$ be a neighbor of $\ell$, which is not a member of his community. If $j \in C$ then $j$ ignores $\ell$'s messages, and there is no incentive to deviate. Suppose $j \in U$. The action of this agent is a (random) function $s^*_j$ of the messages of the leaders in $B_r(j)$, that is $a_j = s^*_j(\pub_{\ell}, \pub_{B_r(j)\backslash \{\ell\}})$. Thus, the probability that $j$ chooses $a_j=\pub_\ell$ is 
$$\Pr{a_j = \pub_\ell} = \Pr{s_j^*(\pub_{\ell}, \pub_{B_r(j)\backslash \{\ell\}}) = \pub_\ell}.$$
Conditioned on sending the message $\pub_\ell=+1$, this probability is
\begin{align*}
\Pr{a_j = \pub_\ell}{\pub_\ell=+1} 
&= \Pr{s_j^*(+1, \pub_{B_r(j)\backslash \{\ell\}}) = +1}{\pub_\ell=+1}\\
&= \Pr{s_j^*(+1, \pub_{B_r(j)\backslash \{\ell\}}) = +1},
\end{align*}
where the second equality is a consequence of the fact that $\pub_\ell$, $\pub_{B_r(j)\backslash \{\ell\}}$ and $s_j^*$ are all independent, as the randomization in the choice of mixed strategy is independent of what occurs in the game.

Now, $s_j^*$ and $\hat s_j^*=\iota(s_j^*)$ have the same distribution by action-symmerty. Thus, and since $s_j^*$ is independent of the public messages,
$$\Pr{a_j = \pub_\ell}{\pub_\ell=+1} = \Pr{\hat s_j^*(+1, \pub_{B_r(j)\backslash \{\ell\}}) = +1}.$$
It follows from the definition of $\iota$ that
\begin{align*}
    \Pr{a_j = \pub_\ell}{\pub_\ell=+1} = \Pr{s_j^*(-1, -\pub_{B_r(j)\backslash \{\ell\}}) = -1}.
\end{align*}
Since $\pub_{B_r(j)\backslash \{\ell\}}$ and $-\pub_{B_r(j)\backslash \{\ell\}}$ have the same distribution (and are independent of $s_j^*$) it follows that
\begin{align*}
    \Pr{a_j = \pub_\ell}{\pub_\ell=+1} = \Pr{s_j^*(-1, \pub_{B_r(j)\backslash \{\ell\}}) = -1},
\end{align*}
which is equal to the probability that $a_j = \pub_\ell$ conditioned on $\pub_\ell=-1$. Thus, the probability that $j$ decides to follow the leader's message is independent of the message, and mixing is a best response for the leader.

Now we argue that agents in $U$ are best responding in the original game. First, notice that non-leaders in $U$ are indifferent between all public messages because their messages are ignored by all other agents. Second, leaders in $U$ are indifferent between all public messages in $\Delta (\{-1, +1\})$ by a similar argument as for non-banished leaders. Furthermore, sending any other message is weakly dominated, since under this deviation the neighbors of the leader will choose their actions uniformly at random. Thus, there exists an obvious lifting of the strategy $s^*_i$ to the original game, in which banished leaders choose a message uniformly at random from $\{-1, +1\}$. It follows immediately from the equilibrium property of $s^*$ that agents in $U$ are best responding on path. This completes the proof of Theorem~\ref{thm:amenable-no-comm}.

\section{Proof of Theorem~\ref{thm:amenable-transitive}}
\label{app:network-symmetric}

An important tool in the proof of Theorem~\ref{thm:amenable-transitive} is \emph{secret-sharing}, which allows us to use public messages to choose for each subgraph $F'$ of radius at most $r$ an independent, uniform $[0,1]$ random variable, in an incentive compatible way. 

Let $\nu$ be the probability measure on the message space $\cM = \cA^\N$ given by choosing each letter independently, from the distribution that assigns probability $1/2$ to $\emptyset$ and probability $1/4$ to both $+1$ and $-1$. Given a message $\pub_i = (\alpha_1,\alpha_2,\ldots)$ chosen from $\nu$, we can map it to a message distributed uniformly on $[0,1]$ by $f(\pub_i) = \sum_{n=1}^\infty 2^{-n}v(\alpha_n)$, where $v(\emptyset) = 0$, $v(-1)=v(+1)=1$. This mapping is invariant to negation, i.e., $f(-\pub_i)=f(\pub_i)$. By a similar construction, we can in fact map $\pub_i$ to a sequence $(\pub_i^1,\pub_i^2,\ldots)$ of i.i.d.\ uniform $[0,1]$ messages. For our purposes we will not need an infinite sequence, and instead will map $\pub_i$ to a tuple $(\pub_i^0,(\pub_i^{F'})_{F'})$, where $F'$ ranges over all radius most $r$ connected subgraphs of $G$ that contain $i$ and at least one other agent. That is, we interpret a message as consisting of an initial number $\pub_i^0$ and another number $\pub_i^{F'}$ for each non-trivial subgraph $i$ participates in.
\begin{lemma}
    \label{lem:secret-sharing}
    Suppose that each agent $i$ sends a public message $\pub_i = (\pub_i^0,(\pub_i^{F'})_{F'})$ as above. For each connected subgraph $F'=(V',E')$ of radius at most $r$ and having at least two vertices, let $Z_{F'} = \sum_{i \in V'}\pub_i^{F'} \mod 1$. Then the collection $Z_{F'}$ is i.i.d.\ uniform on $[0,1]$, and furthermore it is i.i.d.\ uniform on $[0,1]$ even when conditioned on any message $\pub_i$.
\end{lemma}
The proof is immediate, relying on the fact that if $X$ is uniform on $[0,1]$, then so is $X+a \mod 1$ for any $a \in \R$.

Given this lemma, we can define our equilibrium strategies. We fix a connected subgraph $F$ of radius at most $r$, to be determined later. Given a subgraph $F'$, we write $\bar F'$ for the subgraph $F'$, together with the edges in $\partial F'$, and the nodes contained in these edges. I.e., $\bar F'$ is $F'$, expanded to include the nodes connected to it. We call it the expanded graph of $F'$.

%Denote by $\mathcal{K}$ the collection of graphs $\bar F'$ that are isomorphic to $\bar F$.

All agents will send public messages $\pub_i = (\pub_i^0,(\pub_i^{F'})_{F'})$ as above. As in Lemma~\ref{lem:secret-sharing}, let $Z_{F'} = \sum_{i \in V'}\pub_i^{F'} \mod 1$. We fix a probability $p \in [0,1]$, also to be determined later. 

Given this, we define a random set of edges $D = D_1 \cup D_2$, as follows. Let  $F'$ be a subgraph isomorphic to $F$ and such that $\bar F'$ is isomorphic to $\bar F$. We say that $F'$ is \emph{selected} if $Z_{F'} \leq p$. Lemma~\ref{lem:secret-sharing} thus implies that the event that various $F'$s are selected are independent, have probability $p$, and moreover are still independent and have probability $p$ conditioned on any messages sent by an agent $i$.

The set $D_1$ is the union of all the boundaries of the selected subgraphs. The set $D_2$ is the set of all edges of all agents $i$ that do not belong to any selected $F'$. Let $C_1,C_2,\ldots$ be the connected components of the social network graph, with the set of edges $D=D_1\cup D_2$ removed: $(V,E\setminus D)$. We say that $\ell_k \in C_k$ is the leader of $C_k$ if $m_{\ell_k}^0$ maximizes $\{m_i^0\,:\, i \in C_k\}$. I.e., the leader of each community is the member whose first-coordinate message is highest. 

Note that each agent, after observing the public messages, knows which community $C_k$ they are in, and who its leader is, since the radius of communication is $4r$: this radius allows players to know whether they belong to each selected set, as well as  which other selected sets it intersects. This radius also allows leaders to send private messages to all the members of their communities. Indeed, they also send uniform $\{-1,+1\}$ private messages to their community members. A subtlety here is that agents do not know ahead of time whether they will be leaders, and if so who the members of their community will be. Accordingly, each agent sends an independent private message to each potential community: for each $C \subseteq B_r(i)$, agent $i$ sends a independently chosen  recommendation in $\{-1,+1\}$ to all agents in $C$, with this recommendations labeled by the set $C$. Since our message space is rich, there is no issue is accommodating all of this information in the private messages. 

In the last stage, members follow the recommendation sent by their community's leader to their community, and ignore the remaining messages. As in previous constructions, all private messages in an incorrect format are ignored and lead to actions being chosen uniformly at random. Public messages in an unexpected format are likewise ignored and result in the same choice of actions.

Regardless of the choice of $F$ and $p$, this strategy profile is an equilibrium. That there is no incentive to deviate from the public message distribution follows from Lemma~\ref{lem:secret-sharing}. The remaining elements of the profile are as in previously discussed leader equilibria, and hence again admit no profitable deviations. We also note that this strategy profile is action-symmetric. It is also network-symmetric, as its definition does not involve the identity of an agent, or their location in the graph.

The efficiency of this equilibrium depends on the choice of $F$ and $p$. To complete the proof of Theorem~\ref{thm:amenable-transitive}, we show that we can choose $F$ and $p$ to get the required efficiency.

Since the network is an $(\varepsilon,r)$-amenable graph, it can be (deterministically) partitioned into connected communities $\{C_1, ..., C_K\}$ of radius at most $r$, and so that  the set of edges $E'$ connecting different communities is of size $|E'| \leq \varepsilon |E|$. Then, 
$$\frac{|E'|}{|N|} = \frac{1}{|N|}\sum_{k}|\partial C_k| = \frac{1}{|N|}\sum_{k} \frac{|\partial C_k|}{|C_k|}|C_k|,$$
which is the average surface-to-volume ratio weighted by component size. Thus, there exists at least one component $C$ such that
$$\frac{|\partial C|}{|C|} \leq \frac{|E'|}{|N|} \leq \frac{\varepsilon|E|}{|N|} = \varepsilon d_{max},$$
where the equality uses vertex transitivity (so every vertex has degree $d_{max}$) and $E$ is defined as the set of ordered pairs of neighbors $(i, j)$.
We let $F$ be this connected component. %Let $n$ be the number of subgraphs $F'$  such that $\bar F'$ is isomorphic to $\bar F$, and denote by $|F'|$ the number of vertices in $F'$. 

Let $\mathcal{K}$ denote the family of subgraphs $F'$ of $G$ that are isomorphic to $F$, and such that $\bar F'$ is isomorphic to $\bar F$. Let $M_i$ be the number of subgraphs  in $\mathcal{K}$ that contain $i$.
%Let $\mathcal{K}$ denote the family of vertex subsets of $N$, which induce a subgraph isomorphic to $F$ and whose expanded subgraph is isomorphic to $\bar F$.
Let $M_i$ be the number of subgraphs in $\mathcal{K}$ that contain $i$. Since $G$ is $(2r+1)$-locally transitive, every $F' \in \mathcal{K}$ that contains $i$---and its expansion $B_1(i)$---is contained in the ball of radius $2r+1$ around $i$. Hence this number is the same for every vertex; denote it by $M := M_i$. Note that $|\mathcal{K}| = M |N|/|F|$.

Let 
\begin{align*}
    p = \min\{\frac{1}{M}\log\frac{1}{\varepsilon}, 1\}.
\end{align*}
We end the proof of Theorem~\ref{thm:amenable-transitive} by showing that with this choice of $F$ and $p$ we have the required inefficiency.

Let $\mathcal{K}'$ be a random subgraph of $\mathcal{K}$ obtained by selecting each $K \in \mathcal{K}$ independently with probability $p$. The expected number of boundary edges of the selected subgraphs, $
D_1 = \bigcup_{K \in \mathcal{K}'} \partial K$,
is bounded above by the expected number of selected sets times the number of boundary edges in each, that is
\[
\mathbb{E}[|D_1|] \le p \cdot |\mathcal{K}| \cdot |\partial F| 
= p \cdot \frac{|E|\cdot M}{d_{max}\cdot |F|} \cdot |\partial F| 
\leq  p \cdot M \cdot \varepsilon \cdot |E|.
\]

Next, the probability that a vertex $i$ is not contained in any selected $K \in \mathcal{K}'$ equals
\[
\mathbb{P}[\,i \notin K,\, \forall K \in \mathcal{K}'\,] = (1-p)^M.
\]

Given that each vertex has degree $d_{max}$, the expected number of edges whose both endpoints do not belong to any selected $K$ --- denoted by $D_2$ --- is bounded by
\[
\mathbb{E}[|D_2|] \le (1-p)^M \cdot d_{max} \cdot |N| = (1-p)^M \cdot |E|.
\]

Therefore, the expected total number of deleted edges,
$
D= D_1 \cup D_2$,
satisfies
\[
\mathbb{E}[|D|] \le p \cdot M \cdot \varepsilon \cdot |E| 
+ (1-p)^M \cdot |E|.\]

Since $p = \min\{\frac{1}{M}\log\frac{1}{\varepsilon}, 1\}$, we obtain the desired bound on the number of deleted edges. Particularly, if $\frac{1}{M}\log\frac{1}{\varepsilon} \le 1$, then 
\[
(1 - p)^M = \Bigl(1 - \frac{1}{M}\log\frac{1}{\varepsilon}\Bigr)^{M} \le e^{-\log(1/\varepsilon)} = \varepsilon.
\]
Otherwise, $p = 1$, in which case $(1 - p)^M = 0 \le \varepsilon$. This completes the proof of Theorem~\ref{thm:amenable-transitive}.

\begin{remark}
    \cite*{benjamini2008every} apply a similar technique to general graphs (rather than just to $r$-transitive ones, as we do) with the goal of showing that $(\eps,r)$-amenability---i.e., the existence global, deterministic partition into radius at most $r$ connected communities using an $\eps$ fraction of the edges---implies the existence of a local randomized rule that generates the same kind of partition, but (as in our case) using more edges. Their rule is somewhat involved, making it  less plausible as a mechanism for coordination.
\end{remark}

\end{document}